\documentclass[11pt,reqno]{article}
\usepackage[round]{natbib}
\usepackage{amsmath,amsfonts,amsthm}
\usepackage[margin=1in]{geometry}
\usepackage[normalem]{ulem}
\usepackage{amssymb}

\usepackage[hyphens]{url}
\usepackage{graphicx}
\usepackage{verbatim}
\usepackage{enumitem} % to control indention of "itemize" by [leftmargin=0.1in]
\usepackage{color}
\usepackage{tikz}
\usepackage{leftidx} % for left superscript
\usetikzlibrary{decorations.fractals}
\usetikzlibrary{patterns,arrows,shapes,decorations.pathreplacing}
\usepackage{bm}
\usepackage{hyperref}
\numberwithin{equation}{section}

\newtheorem{theorem}{Theorem}[section]

\newtheorem{corollary}{Corollary}[section]
\newtheorem{lemma}{Lemma}[section]
\newtheorem{proposition}{Proposition}[section]

\newtheorem{definition}{Definition}[section]
\newtheorem{remark}{Remark}[section]

\renewcommand{\P}{\mathbb{P}}
\newcommand{\Q}{\mathbb{Q}}
\newcommand{\R}{\mathbb{R}}
\newcommand{\E}{\mathbb{E}}
\newcommand{\cE}{\mathcal{E}}
\newcommand{\N}{\mathbb{N}}
\newcommand{\F}{\mathcal{F}}

\newcommand{\B}{\mathcal{B}}

\newcommand{\T}{\mathcal{T}}

\newcommand{\cU}{\mathcal{U}}
\newcommand{\C}{\mathcal{C}}
\newcommand{\cG}{\mathcal{G}}

\newcommand{\cP}{\mathcal{P}}
\newcommand{\A}{\mathcal{A}}

\newcommand{\eps}{\varepsilon}

\newcommand{\fC}{\mathfrak{C}}

\newcommand{\fP}{\mathfrak{P}}
\newcommand{\pr}{\operatorname{proj}}

\newcommand{\Int}[1]{%
  {\kern0pt#1}^{\mathrm{o}}%
}

\newcommand{\nada}[1]{}
\definecolor{gb}{rgb}{0, 0.2, 0.8}

\title{Optimal Stopping under Model Ambiguity: a Time-Consistent Equilibrium Approach}

\author{Yu-Jui Huang\thanks{
University of Colorado, Department of Applied Mathematics, Boulder, CO 80309-0526, USA, email: \texttt{yujui.huang@colorado.edu}. Partially supported by National Science Foundation (DMS-1715439) and a start-up grant from the University of Colorado (11003573).}
 \and Xiang Yu\thanks{
The Hong Kong Polytechnic University, Department of Applied Mathematics, Hung Hom, Kowloon, Hong Kong, email: \texttt{xiang.yu@polyu.edu.hk}. Supported by the Hong Kong Polytechnic University under grant no. 15304317.}
}
\date{\today}
\begin{document}
\maketitle

\begin{abstract}
An unconventional approach for optimal stopping under model ambiguity is introduced. Besides ambiguity itself, we take into account how {\it ambiguity-averse} an agent is. This inclusion of ambiguity attitude, via an $\alpha$-maxmin nonlinear expectation, renders the stopping problem time-inconsistent. We look for subgame perfect equilibrium stopping policies, formulated as fixed points of an operator. For a one-dimensional diffusion with drift and volatility uncertainty, %{\color{blue}we show that any initial stopping policy will converge to an equilibrium through the proposed iteration scheme and each equilibrium can be obtained via this iteration.}
we show that any initial stopping policy will converge to an equilibrium through a fixed-point iteration.
%every equilibrium can be obtained through a fixed-point iteration. 
This allows us to capture much more diverse behavior, depending on an agent's ambiguity attitude, beyond the standard worst-case (or best-case) analysis. In a concrete example of real options valuation under model ambiguity, all equilibrium stopping policies, as well as the {\it best} one among them, are fully characterized under appropriate conditions. It demonstrates explicitly the effect of ambiguity attitude on decision making: the more ambiguity-averse, the more eager to stop---so as to withdraw from the uncertain environment. The main result hinges on a delicate analysis of continuous sample paths in the canonical space and the capacity theory. To resolve measurability issues, a generalized measurable projection theorem, new to the literature, is also established. 
\end{abstract}

%%%%%%%%%%%%%%%%%%%%%%%%%%%%%%%%%%%%%%%%%%%%%%%%%%%%
%%%%%%%%%%%%%%%%%%%%%%%%%%%%%%%%%%%%%%%%%%%%%%%%%%%%

\section{Introduction}
Decision making under model ambiguity (or, uncertainty) has been extensively studied, {\it dominantly} in the worst-case or the best-case scenario: strategies are found to maximize the worst-case, or the best-case, expected value. In practice, few individuals are so pessimistic (or optimistic) that {\it solely} the least (or the most) favorable situation dictates their behavior. In this paper, a new framework for handling model ambiguity is introduced: {\it ambiguity attitude} of an agent is included as a core ingredient, leading to a more realistic spectrum of behavior.

We focus on optimal stopping. Classically, an agent chooses a stopping time $\tau$ to maximize his expected discounted payoff 
\begin{equation}\label{C}
\E^\P[e^{-r\tau} g(X_\tau)].
\end{equation}
In the face of model ambiguity, the agent, uncertain about the true probability $\P$, can only work with a collection $\cP$ of {\it plausible} probability measures, or {\it priors}, which represent the ambiguity perceived by the agent. This leads to two types of optimal stopping problems. %both of which have been thoroughly investigated. 
The first type---the so-called robust optimal stopping---maximizes the worst-case expected value
\begin{equation}\label{W}
\inf_{\P\in\cP}\E^\P[e^{-r\tau} g(X_\tau)]
\end{equation}
through the choice of $\tau$; see \cite{FrankRiedel2009}, \cite{ErhanYao1, ErhanYao2, BaySong14}, \cite{ChengRiedel13}, and \cite{NutzZhang15}, among many others. The second type, on the other hand, maximizes the best-case expected value
\begin{equation}\label{B}
\sup_{\P\in\cP}\E^\P[e^{-r\tau} g(X_\tau)];
\end{equation}
see e.g. \cite{ErhanYao1, ErhanYao2}, \cite{ETZhang14}, \cite{BelomestnyKra2016}, and \cite{Yao11}. 

What is missing in the above literature is the agent's {\it attitude} towards ambiguity. Even with the same perceived ambiguity $\cP$, different agents may have different levels of ambiguity aversion, as shown empirically in \cite{Curley1989} and  \cite{HeathTversky}.
To differentiate ambiguity attitude from ambiguity itself, \cite{GhirardatoMacMar2004} develop an axiomatic foundation for decision making, leading to the first set of preference models that encode ambiguity attitude. More general models with ambiguity attitude are proposed subsequently in \cite{KlibanoffMM05}, \cite{Chatea}, among others.
%\cite{GhirardatoMacMar2004} recommended the $\alpha$-maxmin preferences via a constant index $\alpha\in [0,1]$ to weight between the worst case and best case objectives. 
In this paper, we incorporate the {\it $\alpha$-maxmin} preference, introduced in Section 6 of \cite{GhirardatoMacMar2004}, into the optimal stopping framework: the agent intends to maximize
\begin{align}\label{fucnal'}
\alpha \inf_{\P\in\cP} \E^{\P}[e^{-r\tau}g(X_\tau)] + (1-\alpha) \sup_{\P\in\cP} \E^{\P}[e^{-r\tau}g(X_\tau)],
\end{align}
where $\alpha\in [0,1]$ is a given constant that reflects the level of ambiguity aversion of the agent. Here, both {\it ambiguity} and {\it ambiguity attitude} are captured, by $\cP$ and $\alpha$, respectively. The case $\alpha=1$ amounts to the standard worst-case analysis, reflecting extreme aversion to ambiguity. The other extreme $\alpha=0$ depicts a purely ambiguity-loving agent who cares only about the best-case value. 
Note that  an agent's preference is fully characterized by an $\alpha$-maxmin objective, such as \eqref{fucnal'}, as long as six basic axioms are satisfied; see Section 6 of \cite{GhirardatoMacMar2004}.
The goal of this paper is to investigate stopping behavior under the $\alpha$-maxmin objective \eqref{fucnal'}. 
%Other more complicated ways to capture heterogeneous ambiguity attitude, such as smooth ambiguity preference and neo-additive capacity model, can be found in \cite{KlibanoffMM05} and \cite{Chatea} and etc.

A distinctive challenge in solving \eqref{fucnal'} is {\it time inconsistency}: an optimal strategy we find today may no longer be optimal at future dates. That is, our future selves may very well deviate from the optimal strategy we set out to employ today. %Mathematically speaking, the problem \eqref{fucnal'} is thereby more descriptive than normative as no dynamic global solution is guaranteed. 
Consequently, finding an optimal stopping time, the ultimate goal in the standard literature, is not meaningful here. 

Note that neither the classical problem \eqref{C} nor the worst-case and best-case problems \eqref{W} and \eqref{B} suffers the issue of time inconsistency. Indeed, time consistency of \eqref{C} simply boils down to the tower property of conditional expectations. While time consistency is generally in question under nonlinear expectations, \cite{ES03} show that, for the special cases \eqref{W} and \eqref{B}, tower property still holds if the set of priors is {\it rectanguler}, i.e., closed under conditioning and stable under pasting. Similar conditions are discovered independently, and further refined under great generality, in the literature of mathematical finance; see e.g. \cite{NvH13}, \cite{BaySong14}, \cite{ETZhang14}, and \cite{NutzZhang15}. All the developments ensure certain tower property for the nonlinear expectation \eqref{W} or \eqref{B}, so that time consistency follows. %the existence of global optimal stopping time can be achieved in the time consistent framework. %Classical methods, such as dynamic programming and the martingale approach, can then be invoked to find an optimal stopping time. 
By contrast, an $\alpha$-maxmin objective, such as \eqref{fucnal'}, does not uphold time consistency, even when the set of priors is rectangular. This is demonstrated in Section 7 of \cite{Schroder} and %a simple two-period model in 
Section 2 of \cite{BLR16}. Time inconsistency is a genuine difficulty for \eqref{fucnal'}.

As proposed in \cite{Stroz1956myopia}, a sensible way to deal with time inconsistency is {\it consistent planning}:  knowing that his future selves may overturn his current plan, the agent selects the best present action taking the future disobedience as a constraint. Assuming that every future self will reason in the same way, the resulting strategy is a  (subgame perfect) equilibrium, from which no future self has incentive to deviate. How such strategies can be precisely formulated and obtained has been a long-standing problem. In response to this, \cite{HN18} develop an iterative approach to finding equilibria for time-inconsistent stopping problems. %: equilibrium strategies, characterized as fixed points of an operator, can be found conveniently via fixed-point iterations. 
It has been applied to stopping under non-exponential discounting (\cite{HN18}, \cite{HZ17-discrete, HZ17-continuous}) and probability distortion (\cite*{HNHZ18}). 

In this paper, we extend the framework of \cite{HZ17-continuous} to account for the $\alpha$-maxmin objective \eqref{fucnal'}. Equilibrium stopping policies are characterized as fixed-points of an operator $\Theta$, defined in \eqref{opertheta} below. The central question is whether equilibria can be found via fixed-point iterations. We take up a strong formulation of model ambiguity, where the drift and volatility coefficients of a one-dimensional diffusion $X$ are only assumed to satisfy certain Lipschitz and linear growth conditions, and are otherwise unknown. As shown in Lemmas~\ref{lem:rc} and \ref{lem:regular}, the resulting collection of priors $\cP$ is relatively compact and $X$ is a regular diffusion under any $\P\in\cP$. The regularity of $X$ immediately yields the convergence of any fixed-point iteration; see Proposition~\ref{limitstopping}. To show that the limit of a fixed-point iteration is indeed an equilibrium, appropriate convergence of stopping times and the values at stopping, uniform in $\P\in\cP$, is required. Such uniform convergence is carefully established in Lemma~\ref{lem:in capacity}, relying crucially on both the relative compactness of $\cP$ and the regularity of $X$. All this leads to Theorem~\ref{thm:main}, the main result of this paper: {\it every} equilibrium can be found via a fixed-point iteration. 

Our framework, in particular, sheds new light on real options valuation.
%To demonstrate our theoretical results, we look into a real options valuation problem in detail. 
The essence of real options valuation is the use of financial pricing techniques to evaluate the right, but not the obligation, to undertake certain capital investment project. 
%, Traditionally, this boils down to an optimal stopping problem under a risk-neutral measure, whose solution dictates optimal timing or scheduling of investment outlays. 
By nature, real options valuation may suffer model ambiguity {\it more} severely than pricing a typical financial option: as the underlying asset of a real option is mostly neither tradable nor fully observable, determining its dynamics relies largely on an agent's estimate. This often leads to %a collection of plausible risk-neutral measures and a corresponding
an interval of plausible values of a real option. By incorporating the $\alpha$-maxmin preference, the multiple plausible values turns into a single one, i.e., the convex combination of the least and the best values, as in \eqref{fucnal'}. This facilitates decisions making: one compares this single value and the value of immediate stopping, to decide whether a project should be postponed or initiated. While the involved stopping problem is now time-inconsistent, %, inapproachable by standard techniques of real options valuation, 
the methodology we develop comes into play to locate (time-consistent) equilibrium strategies. 

In particular, in the uncertain volatility model introduced by \cite{ALP95} and \cite{Lyons95}, when the payoff function of a real option is of the put option type, we provide complete characterizations of not only {\it all} the equilibrium strategies, but also the {\it best} one among them, under appropriate conditions; see Proposition~\ref{characterization} and Theorem~\ref{thm:optimal E}. It demonstrates explicitly the effect of ambiguity attitude: the more ambiguity-averse, the more eager to stop---so as to withdraw from the uncertain environment. 

In summary, the main contributions of this paper are as follows: 
\begin{itemize}[leftmargin=*]
\item[(i)] To the best of our knowledge, this is the first paper that resolves the time-inconsistent stopping problem under the $\alpha$-maxmin preference\footnote{There is a related stopping problem, with drift uncertainty only, introduced in \cite{Schroder} under the $\alpha$-maxmin preference. However, due to the time inconsistency involved, the stopping problem was not solved therein, except for the extreme cases $\alpha=1$ and $\alpha=0$ (i.e., the usual worst case and best case again).}. %Specifically, we show that any (time-consistent) equilibrium strategy can be found via a fixed-point iteration. 
This allows us to go beyond the standard worst-case (or best-case) analysis under model ambiguity, and capture a more realistic spectrum of behavior. We stress that our collection of priors $\mathcal{P}$ is only assumed to be measurable; the ``rectangular'' condition,  imposed widely in the literature, is not required; see Remark~\ref{rem:rectangular} for details. 

\item [(ii)]
Whereas time-inconsistent stopping behavior has been widely investigated, it has been ascribed mostly to {\it non-exponential discounting}, {\it probability distortion}, {\it dependence on initial data}, or {\it nonlinearity in expected rewards}; see e.g., \cite{grenadier2007investment}, \cite{Barberis12}, \cite{XZ13}, \cite{ebert2015until}, \cite{CL18,CL20}. This paper enriches research on time-inconsistent stopping, by focusing on {\it ambiguity aversion}, a cause of time inconsistency that has only been slightly discussed in the literature.

%Under model ambiguity, we no longer restrict ourselves to the worst-case (or the best-case) analysis, as opposed to the majority of the literature. Besides ambiguity itself, an agent's ambiguity aversion is included through the $\alpha$-maxmin preference. 

\item [(iii)]
Our framework provides a new approach for real options valuation. Taking ambiguity attitude into account, via the $\alpha$-maxmin preference, facilitates decision making under model ambiguity (as discussed above), but it also renders the stopping problem time-inconsistent.  
%reduces multiple plausible values of a real option to one, which facilitates decision making. 
The methodology we develop particularly resolves this time-inconsistent problem, %the well-known time inconsistent open problem in robust valuation procedure, 
allowing us to take full advantage of including ambiguity attitude in decision making. 
%Second, we demonstrate explicitly how ambiguity aversion significantly affects agents' choice of equilibrium stopping policies. 
%This may shed new light on the diverse strategies among investors in face of similar, yet uncertain, investment opportunities. 
%Indeed, the release of extra dimension of attitude yields diversity of the derived stopping decisions, which matches with real life evidences that some entrepreneurs are fully confident of their business ideas and pursue the investment in new and emerging industry while many others will not. 

\item[(iv)] Extending the iterative approach from \cite{HN18} to our multiple-prior setting is nontrivial. It demands several convergence results related to stopping times, uniform across all priors, which are established by a detailed analysis of sample paths and a careful use of the capacity theory; see Lemma~\ref{lem:in capacity} and its proof in Appendix~\ref{sec:proof of lem:in capacity}. 

Moreover, we highlight the new measurable projection theorem established in Theorem~\ref{thm:projection} below, which may have many potential applications in stochastic analysis beyond  this paper. %A measurable projection theorem involves the product of two measurable spaces, and studies whether the projection of a measurable set in the product space is still measurable. 
Classical measurable projection theorems all require one of the spaces involved to be a Borel space endowed with the Borel $\sigma$-algebra. Theorem~\ref{thm:projection}, by contrast, allows for {\it any} general measurable spaces. 
In our multiple-prior setting, for the fixed-point operator $\Theta$ to be well-defined, Borel measurability, used in the single-prior framework of \cite{HNHZ18} and \cite{HZ17-continuous}, is no longer adequate, and the more general universal measurability is needed; see Section~\ref{why use cU} for detailed explanations. Showing that the objective function \eqref{fucnal'} is universally measurable then demands Theorem~\ref{thm:projection}, which does not require specific Borel structure; see Lemma~\ref{measuretheta} and Remark~\ref{clremk}. 
%\item[(iii)]  
% If the risk neutral valuation is employed in a robust framework to compare the momentary payoff and the discounted future payoff, we may encounter a dilemma that the current payoff sometimes stays between the worst case and best case future payoffs. No sound conclusion, namely stop or continue, can be determined if the investor is not extremely pessimistic. Some previous work only analyze reasonable actions case by case. By means of our optimal stopping under $\alpha$-maxmin preference, we gladly discover that the index $\alpha$ provides a unified criterion so that payoffs become comparable and business actions can be made even when the underlying model is not specified.    
\end{itemize}

The rest of the paper is organized as follows. Section \ref{sec:model} introduces the general set-up, including the formulations of the time-inconsistent stopping problem under model ambiguity and the corresponding fixed-point operator. In Section \ref{sec:converg}, under both drift and volatility uncertainty of a one-dimensional diffusion, we show that every equilibrium stopping policy is the limit of a fixed-point iteration. 
Section \ref{sec:example} studies a concrete real options valuation problem; all equilibrium stopping policies, and the best one among them, are explicitly characterized under appropriate conditions. Section \ref{sec:proj} is devoted to the derivation of a new, generalized measurable projection theorem, which is required in Section~\ref{sec:model}. Appendix~\ref{sec:proof of lem:in capacity} presents the technical proof of Lemma~\ref{lem:in capacity}.

\section{The Set-up}\label{sec:model}
%We start with an abstract set up to account for decision maker's ambiguity on the underlying models. In particular, we are interested in the scenario of volatility uncertainty so that the family of unspecified model priors is not necessarily dominated by one reference probability measure. To formulate the robust optimal stopping problem in a general framework, 

For any Polish space $M$, we denote by $\B(M)$ the Borel $\sigma$-algebra of $M$, and by $\cU(M)$ the $\sigma$-algebra consisting of all universally measurable sets in $M$. Let $\mathfrak{P}(M)$ be the set of all probability measures on $(M,\B(M))$. Each $\P\in\mathfrak P(M)$ can be uniquely extended to $\cU(M)$, and we do not distinguish between $\P$ and such an extension. In this paper, we will particularly take $M$ to be either $\R^d$ or the canonical spaces of continuous paths, i.e., $\Omega$ and $\Omega_t$ to be defined below.  
%$\Omega := C([0,\infty);\R^d)$, , or $\mathbb{R}^d$, such that $\cU(\Omega)$ and $\cU(\mathbb{R}^d)$ are defined using all probability measures in $\mathfrak{P}(\Omega)$ and $\mathfrak{P}(\mathbb{R}^d)$ respectively. 

Consider the canonical space $\Omega := C([0,\infty);\R^d)$. %, i.e., the space of continuous paths starting with $0\in \R^d$. 
For each $t>0$, we define $\Omega_t := C([0,t];\R^d)$. Let $B$ denote the canonical process $B_t(\omega) := \omega_t$ for all $\omega\in\Omega$, and let $\mathbb{F}^B= (\F^B_t)_{t\ge 0}$ be the natural filtration generated by $B$. Recall that \begin{equation}\label{F^X=B}
\F^B_t = \B(\Omega_t),\ \ \forall t\ge 0\quad \hbox{and}\quad \F^B_\infty = \B(\Omega).
\end{equation}
For each $\P\in\fP(\Omega)$, let $\mathbb{F}^\P= (\F^\P_t)_{t\ge 0}$ be the $\P$-augmentation of $\mathbb{F}^B$. We then define the universal filtration $\mathbb{F} = (\F_t)_{t\ge 0}$ by
\[
\F_t := %\cG_t := 
\bigcap_{\P\in \fP(\Omega)} \F_t^{\P},\quad \forall t\ge 0.
\]
Note that $\mathbb F$ is right-continuous, thanks to the right-continuity of $\mathbb F^\P$ for all $\P\in  \fP(\Omega)$. Moreover, 
\begin{equation}\label{G=U}
\F_t = \cU(\Omega_t),\ \ \forall t\ge 0\quad \hbox{and}\quad \F_\infty = \cU(\Omega).
\end{equation}
%Let $\mathbb{F}= \{\F_t\}_{t\ge 0}$ be defined as $\F_t := \cG_{t+}$ for all $t\ge 0$. 
We denote by $\T$ the set of all $\mathbb F$-stopping times.

Let us introduce a general, albeit time-homogeneous, formulation of model ambiguity. For any $x\in \R^d$, consider
\[
\Omega^x := \{\omega\in\Omega: \omega_0=x\},
\]
and let 
\begin{equation}\label{cP}
\cP(x) \subseteq \{\P\in\fP(\Omega) : \P(\Omega^x) =1\ \hbox{and}\ B\ \hbox{is strong Markov under $\P$}\} %\quad \forall x\in\R,
\end{equation}
denote the set of {\it priors} of an agent at the state $x\in\R^d$. That is, %$\cP(x)$ represents the beliefs of the agent at $x\in\R^d$: 
every $\P\in\cP(x)$ is believed by the agent to be a possibly true description of how the process $B$ will evolve, given that its current value is $x\in\R^d$. Note that $\cP(x)$ is {\it not} necessarily dominated by some reference probability $\P^*$ with respect to which all $\P\in\cP(x)$ are absolutely continuous. In other words, some elements in $\cP(x)$ may be mutually singular, which in particular covers the case of volatility uncertainty in a diffusion model of $B$; see Section~\ref{sec:example} for a detailed example. 

%%%%%%%%%%%%%%%%%%%%%%%%%%%%%%%%%%5

\subsection{The $\alpha$-maxmin Objective and Time Inconsistency}\label{subsec:alpha}
Consider a payoff function $g:\R^d\to \R$. %, assumed to be continuous. 
An agent, with discount rate $r>0$, intends to maximize $\E^{\P}[e^{-r\tau}g(B_\tau)]$ by choosing an appropriate $\tau\in\T$, subject to the uncertainty $\P\in\cP(x)$ at the current state $x\in\R^d$. Such a stopping problem has been substantially studied, {\it yet} almost always in the worst-case (or best-case) scenario. Stated in the current setting, the literature is focused on finding $\tau^*\in\T$ that maximizes the worst-case (or best-case) value, i.e.,
\begin{equation}\label{worst case}
\inf_{\P\in\cP(x)} \E^{\P}[e^{-r\tau}g(B_\tau)]\quad \hbox{or}\quad \sup_{\P\in\cP(x)} \E^{\P}[e^{-r\tau}g(B_\tau)].
\end{equation}
A large number of references can be found in the introduction. 
%Studies along this line include \cite{Smith2002}, \cite{FrankRiedel2009}, \cite{ErhanYao1, ErhanYao2, BaySong14}, \cite{ChengRiedel13}, \cite{ETZhang14}, \cite{NutzZhang15}, and \cite{BelomestnyKra2016}, among many others. 

Practical decision making, however, is much more complicated than the worst-case (or best-case) analysis. What is missing in \eqref{worst case} is the agent's {\it attitude} towards ambiguity: even with the same perceived ambiguity $\cP(x)$, different agents may have different levels of ambiguity aversion. As shown empirically in \cite{Curley1989} and \cite{HeathTversky}, ambiguity attitude is heterogeneous among individuals: some can be much less ambiguity-averse than others under various circumstances. To accommodate ambiguity attitude, a general model of utility maximization has been developed in \cite{GhirardatoMacMar2004} and \cite{KlibanoffMM05}. In particular, the {\it $\alpha$-maxmin} preference, a popular, straightforward version of the general model, stipulates that, at the current state $x\in\R^d$, the agent maximizes
%expected payoff under model ambiguity and ambiguity attitude, associated with the $\mathbb F$-stopping time $\tau\in\T$, is characterized by the $\alpha$-maxmin functional
\begin{align}\label{fucnal}
\alpha \inf_{\P\in\cP(x)} \E^{\P}[e^{-r\tau}g(B_\tau)] + (1-\alpha) \sup_{\P\in\cP(x)} \E^{\P}[e^{-r\tau}g(B_\tau)],
\end{align}
where $\alpha\in [0,1]$ is a given constant that reflects the level of ambiguity aversion of the agent. Here, both ambiguity and ambiguity attitude are captured, by $\cP(x)$ and $\alpha$, respectively. The case $\alpha=1$  (resp. $\alpha=0$) corresponds to the standard worst-case (resp. best-case) problem, reflecting extreme aversion to (resp. desire for) ambiguity.  It is the scope of this paper to investigate the diverse stopping behavior between these two extremes, i.e., for any $\alpha\in [0,1]$. 

When solving the problem
\begin{equation}\label{supJ}
\sup_{\tau\in\T} \bigg(\alpha \inf_{\P\in\cP(x)} \E^{\P}[e^{-r\tau}g(B_\tau)] + (1-\alpha) \sup_{\P\in\cP(x)} \E^{\P}[e^{-r\tau}g(B_\tau)]\bigg),
\end{equation}
the issue of {\it time inconsistency} arises: an optimal strategy we find today may no longer be optimal at future dates. 
Specifically, suppose an optimal stopping time $\widetilde \tau_x\in\T$ exists for \eqref{supJ}, for all $x\in\R^d$. The problem \eqref{supJ} is said to be time-consistent if for any $x\in\R^d$ and $t\ge 0$,
\begin{align}\label{consistent condition}
\widetilde\tau_x(\omega)=t+\widetilde\tau_{B_t}(\omega)\quad \text{for}\ \omega\in\{\tau\geq t\}\quad \P\text{-a.s.},\qquad \forall \P\in\cP(x).
\end{align}
If the above condition fails to hold, \eqref{supJ} is said to be time-inconsistent. 

A critical condition time consistency hinges on is the tower property of conditional expectations. In the worst case (or best-case) scenario, \cite{ES03} shows that time consistency holds if the set of priors is {\it rectanguler}, i.e., closed under conditioning and stable under pasting. Similar conditions are proposed under great generality in the literature of mathematical finance; see  \cite{NvH13}, \cite{BaySong14}, \cite{ETZhang14}, and \cite{NutzZhang15}. The technical endeavor in these works ensures certain tower property of nonlinear conditional expectations in the form of \eqref{worst case}, so that time consistency follows. 

By contrast, time inconsistency is inherent in \eqref{supJ}. Even when the set of priors is rectangular, 
an $\alpha$-maxmin objective does not uphold time consistency. This is demonstrated in detail through an optimal stopping problem in Section 7 of \cite{Schroder}, as well as in a simple two-period model in Section 2 of \cite*{BLR16}.

\begin{remark}
When the set of priors is dominated (i.e., a reference measure $\mathbb{P}$ exists such that $\Q\ll \P$ for each prior $\Q$), ambiguity aversion can be modeled via an entropic penalty term (``minimal penalty'' in the literature). As specified in \cite{BayKaYao}, the expected payoff under each prior $\Q$ is penalized by its distance from $\P$, and the resulting robust optimal stopping problem is still time-consistent. Our framework differs from this in two ways. 
%\begin{itemize}
%\item [(i)] 

Technically, we do not require the set of priors to be dominated, thereby allowing for more general forms of model ambiguity, such as volatility uncertainty (under which priors can be mutually singular). 
%\item [(ii)] 
Economically, the entropic penalization assumes (i) an agent already has a specific belief (the reference measure $\P$) and (ii) he is ambiguity-averse---so that the expected payoff under a prior $\Q$ is weighted by how much $\Q$ differs from his belief $\P$. This paper assumes neither (i) nor (ii). Through the $\alpha$-maxmin objective, we cover a wide range of ambiguity attitude---from extreme ambiguity aversion ($\alpha=1$) to pure ambiguity loving ($\alpha=0$)---without the need of a specific belief form the agent.
%prescribes an agent's ambiguity attitude in a fixed, less flexible way: The agent is ambiguity-averse and already has a firm belief of ...
% that describes the distance between the reference measure and the prior distribution. The robust optimal stopping problem under this type of preference, see \cite{BayKaYao}, is still time consistent. 
%Moreover, the entropic penalty formulation actually characterizes how much the agent would take into account the model ambiguity in the decision making, i.e., how robust the decision making is. In contrast, the ambiguity attitude via $\alpha$-maxmin is defined in a way that the decision making is fully robust and the optimal decision is based on how the agent views all payoffs under these distributions. Under the umbrella of $\alpha$-maxmin preference, some agents can be less ambiguity averse than others in the sense these agents care less about the worst case payoff that may hinder their decision making. To incorporate the latter ambiguity aversion attitude by agents and to consider possible volatility uncertainty in applications, we choose to work with the $\alpha$-maxmin preference, which is the source of time inconsistency in the present paper.
%\end{itemize}
\end{remark}

To deal with time inconsistency, we follow {\it consistent planning} proposed in \cite{Stroz1956myopia}:  one takes into account the potential disobedience of his future selves, and selects the best present action in response to that. If every future self will reason in the same way, the resulting strategy will be a (subgame perfect) equilibrium, from which no future self has incentive to deviate. How to precisely formulate and locate such strategies has been a long-standing challenge.  
In the context of time-inconsistent stopping, \cite{HN18} develop a versatile iterative approach: equilibrium strategies, formulated as fixed points of an operator, can be found conveniently via fixed-point iterations. It has been applied successfully to optimal stopping under non-exponential discounting (\cite{HN18}, \cite{HZ17-discrete, HZ17-continuous}) and probability distortion (\cite*{HNHZ18}). We will now extend this iteration approach further to account for model ambiguity and ambiguity attitude.  

%%%%%%%%%%%%%%%%%%%%%%%%

\subsection{Consistent Planning}
%{\color{blue}The majority of the optimal stopping time in the existing work on time consistent robust optimal stopping (with drift uncertainty or volatility uncertainty) in the Markovian framework fit into the type of first hitting time. This motivates us to consider the first hitting time in our time inconsistent setting with model ambiguity described by \eqref{cP}. More importantly, the first hitting time allows us to pass the stopping policy to its associated set, which is technically convenient for us to define and verify the convergence of our proposed iteration under model ambiguity. %Thanks to the time-homogeneous setup of model ambiguity in \eqref{cP}, Starting from this point on, let us focus on hitting times to regions in $\R^d$, instead of dealing with all general stopping times.} 
Thanks to the time-homogeneous Markovian setup in \eqref{cP}, we assume that an agent decides to stop or to continue depending on his current state $x\in\R^d$. That is, the agent chooses some $R\in \cU(\R^d)$, and stops at the moment 
\begin{equation}\label{tauR}
\tau_R:= \inf\{t\ge 0: B_t\in R\}. %\quad \hbox{for}\ R\in\cU(\R^d),
\end{equation}
This corresponds to a {\it pure strategy} in game theory. While {\it mixed strategies} could also be considered here, we will leave it for future research. For convenience, we will often call $R\in \cU(\R^d)$ a {\it stopping policy} in the rest of the paper. 

It is worth noting that we do not restrict ourselves to Borel measurable %or analytic stopping 
policies $R$, but allow for universally measurable ones.  
Such generality is essential to our subsequent fixed-point formulation \eqref{iteration}; see the detailed explanations in Section~\ref{why use cU} below. 

To carry out consistent planning proposed in \cite{Stroz1956myopia}, we follow the game-theoretic formulation in Section 2.1 of \cite{HZ17-continuous} (in line with Section 3.1 of \cite{HN18}). Suppose that the agent initially planned to take $R\in \cU(\R^d)$ as his stopping policy. Given the current state $x\in\R^d$, the agent carries out the game-theoretic reasoning: ``assuming that all my future selves will follow $R\in\cU(\R^d)$, what is the best stopping strategy today in response to that?'' The agent today has only two possible actions: stopping and continuation. If he stops, he gets $g(x)$ right away; if he continues, given that all his future selves will follow $R\in\cU(\R^d)$, he will eventually stop at the moment 
\begin{equation}\label{rhoR}
\rho_R:= \inf\{t> 0: B_t\in R\}.
\end{equation}
Given that the agent's ambiguity attitude is characterized by $\alpha\in [0,1]$, this leads to the $\alpha$-maxmin expected payoff
\begin{equation}\label{J}
J(x,R) := \alpha \inf_{\P\in\cP(x)} \E^{\P}[e^{-r\rho_R}g(B_{\rho_R})] + (1-\alpha) \sup_{\P\in\cP(x)} \E^{\P}[e^{-r\rho_R}g(B_{\rho_R})].
\end{equation}

\begin{remark}
The fact ``$\rho_R\in \T$'' can be proved in the same way as in Lemma~\ref{lem:in T}. Note the subtle difference between $\tau_R$ and $\rho_R$: the former involves ``$t\ge 0$'', while the latter ``$t>0$''. In \eqref{J}, as this is the case where the agent at $x\in\R^d$ chooses to continue (without regard to whether $x\in R$), the stopping time in effect is $\rho_R$, not $\tau_R$.
\end{remark}

%\begin{remark}
In \eqref{J}, we allow $\rho_R$ to take the value $\infty$: if $\rho_R(\omega)=\infty$, we define
\begin{equation}\label{at infty}
e^{-r\rho_R}g(B_{\rho_R})(\omega) :=  \limsup_{t\to\infty} e^{-rt}g(B_t)(\omega).
\end{equation}
This is in line with Appendix D of \cite{KS-book-98}.
%so that the expectations in \eqref{J} can be evaluated. 
%\end{remark}
%Moreover, to ensure that $J(x,R)$ in \eqref{J} is well-defined, we impose throughout the paper that for any $x\in\R^d$ and $t\ge 0$,
%%\begin{assumption}\label{standing}
%%For any $x\in\R^d$, 
%\[
%\sup_{\P\in\cP(x)}\E^\P\left[e^{-rt} |g(B_t)|\right] <\infty,
%%\sup_{\P\in\cP(x)}\E^\P\left[\sup_{0\le t\le\infty} e^{-rt} |g(B_t)|\right] <\infty,
%\]
%where we interpret $e^{-r\cdot \infty }|g(B_\infty)|:= \limsup_{t\to\infty} e^{-rt}|g(B_t)|$, similarly to \eqref{at infty}. 
%\end{assumption}

To find the best stopping policy for today, in response to future selves following $R\in\cU(\R^d)$, the agent simply compares the payoffs $g(x)$ and $J(x, R)$. This leads to % we obtain the best stopping strategy for today as
\begin{align}\label{opertheta}
\Theta(R) := S_R\cup(I_R\cap R),
\end{align}
where we define
\begin{equation}\label{regions}
\begin{split}
S_R &:= \{x\in\R^d : g(x) > J(x,R)\},\\
I_R &:= \{x\in\R^d : g(x) = J(x,R)\},\\
C_R &:= \{x\in\R^d : g(x) < J(x,R)\}.
\end{split}
\end{equation}
Here, $S_R$, $I_R$, and $C_R$ are called the stopping region, the indifference region, and the continuation region, respectively. In particular, on $I_R$, the agent is indifferent between stopping and continuation as they yield the same payoff. There is then no incentive for the agent to deviate from the original stopping policy $R\in\cU(\R^d)$. This gives rise to the term $I_R\cap R$ in \eqref{opertheta}.\footnote{A similar formulation can be found in Section 2.1 of \cite{HZ17-continuous}.}

It is of interest to determine whether this new region $\Theta(R)\subseteq \R^d$, obtained from the original stopping policy $R\in\cU(\R^d)$, %via the above game-theoretic reasoning, 
is again a stopping policy, i.e., $\Theta(R)\in\cU(\R^d)$. If this is true, it will facilitate the fixed-point formulation in \eqref{iteration} below. In view of \eqref{regions}, whether $\Theta(R)$ belongs to $\cU(\R^d)$ hinges on the measurability of $x\mapsto J(x,R)$, which will now be investigated.

%%%%%%%%%%%%%%%%%%%%%%%%%%%%%%%

\subsection{Measurability}
First, we show that $\tau_R$, with $R\in \cU(\R^d)$, is a well-defined stopping time. Note that this result does not follow from the standard debut theorem; see Remark~\ref{rem:debut} below. %the measurability of $\tau_R$ is nontrivial. 
%Indeed, if we had $R\in\B(\R^d)$, the classical debut theorem (see e.g. Theorem 2.1 in \cite{Bass10}) would imply that for any $t\ge 0$, $\{\tau_R\le t\}\in\F^\P_t$ for all $\P\in\fP(\Omega)$, which readily yields $\tau_R\in\T$. As the next result shows, much more care needs to be taken for $R\in \cU(\R^d)$. 

\begin{lemma}\label{lem:in T}
For any $R\in\cU(\R^d)$, $\tau_R$ in \eqref{tauR} belongs to $\T$.  
%\begin{equation}\label{rhoR}
%\tau_R:= \inf\{t\ge 0: B_t\in R\}\quad \hbox{and}\quad \rho_R:= \inf\{t> 0: B_t\in R\}
%\end{equation}
\end{lemma}

\begin{proof}
For each fixed $s\geq 0$, $\omega\mapsto B_s(\omega)$ is by definition $\F^B_s$-measurable. Thanks to \eqref{F^X=B} and Corollary 7.44.1 of \cite{BS-book-78}, we have the relation
\begin{equation}\label{each s}
A\in\cU(\R^d)\quad \hbox{implies}\quad (B_s)^{-1}(A)\in \cU(\Omega_s) = \F_s \subseteq \F^\P_s\ \hbox{for all}\ \P\in\fP(\Omega),
\end{equation}
where the equality follows from \eqref{G=U}.

Fix $t>0$, by the right continuity of $B$, we can construct, as in Proposition 1.1.13 of \cite{KS-book-91}, a sequence of discretized processes $\{B^{(n)}\}_{n\in\N}$, which satisfy $B^{(n)}_s(\omega)\to B_s(\omega)$ for all $(s,\omega)\in[0,t]\times\Omega$. For each $\P\in\fP(\Omega)$, by using \eqref{each s}, the constructed map $(s,\omega)\mapsto B^{(n)}_s(\omega)$ from $([0,t]\times \Omega, \B([0,t])\times \F^\P_t)$ to $(\R^d,\cU(\R^d))$ is measurable, for all $n\in\N$. As $n\to\infty$, we conclude that the map $(s,\omega)\mapsto B_s(\omega)$, again from $([0,t]\times \Omega, \B([0,t])\times \F^\P_t)$ to $(\R^d,\cU(\R^d))$, is also measurable. Given $R\in\cU(\R^d)$, it follows that
\begin{equation}\label{Gamma_t}
\Gamma_t := \{(s,\omega)\in [0,t)\times \Omega: B_s(\omega)\in R\} \in\B([0,t])\times \F^\P_t,\quad \forall t\ge 0\ {and}\ \P\in\fP(\Omega).
\end{equation}
Now, for each $\P\in\fP(\Omega)$, thanks to Theorem I.4.14 of \cite{RY-book-99} and \eqref{Gamma_t}, we have $\{\tau_R < t \} = \text{proj}_{\Omega}(\Gamma_t) \in\F^\P_t$, for all $t\ge 0$. Thus, $\{\tau_R < t \}\in \bigcap_{\P\in \fP(\Omega)} \F^\P_t = \F_t$, for all $t\ge 0$. As $\mathbb F=(\F_t)_{t\ge 0}$ is by construction right-continuous, $\tau_R$ is an $\mathbb F$-stopping time, i.e., $\tau_R\in\T$. %It can be proved similarly that $\rho_R$ is an $\mathbb F$-stopping time.
\end{proof}
	
\begin{remark}\label{rem:debut}
If we had $R\in\B(\R^d)$, $\tau_R\in\T$ would follow directly from the debut theorem (see Theorem I.4.15 in \cite{RY-book-99} or Theorem 2.1 in \cite{Bass10}). Indeed, the theorem would imply that for any $t\ge 0$, $\{\tau_R\le t\}\in\F^\P_t$ for all $\P\in\fP(\Omega)$, which readily yields $\tau_R\in\T$. The debut theorem, noticeably, requires progressive measurability of $\{(s,\omega):B_s(\omega)\in R\}$, i.e.
\begin{equation*}%\label{debut cond}
\{(s,\omega): 0\le s\le t,\ B_s(\omega)\in R\}\in\B([0,t])\times\F_t,\quad\forall t\ge 0.
\end{equation*}
This holds trivially for $R\in\B(\R^d)$ (by the fact that $B$ is progressively measurable), but is questionable for $R\in\cU(\R^d)$. %Hence, it is unclear if the debut theorem can be applied with $R\in\cU(\R^d)$. 
Without using the debut theorem, Lemma~\ref{lem:in T} proves $\tau_R\in\T$ generally for all $R\in\cU(\R^d)$. %, under additional technicalities.  
%As the next result shows, much more care needs to be taken for $R\in \cU(\R^d)$. 
\end{remark}	
	
\begin{remark}\label{rem:two Borel concepts}
There are two related, but distinct, ``Borel measurability'' of $\tau_R$: %(and $\rho_R$).
\begin{itemize}
\item [1.] $\tau_R$ is a Borel measurable random variable, i.e. $\{\tau_R\le t\}\in\mathcal F^B_\infty$ for all $t\ge0$. % is $\mathcal F^B_\infty$-measurable). 
\item [2.] $\tau_R$ is an $(\mathcal F^B_t)$-stopping time, i.e. $\{\tau_R\le t\}\in\mathcal F^B_t$ for all $t\ge 0$. 
\end{itemize} 
For $R\in\B(\R^d)$, $\tau_R$ is always a Borel measurable random variable, thanks to Galmarino's test (see e.g. p.46, Exercise 4.21 in \cite{RY-book-99}). However, $\tau_R$ is not necessarily an $(\mathcal F^B_t)$-stopping time, for instance, when $R$ is open; see p.43, Proposition 4.6 and the discussion below it in \cite{RY-book-99}. In general, $\tau_R$, even with $R\in\B(\R^d)$, is only known to be an $(\mathcal F_t)$-stopping time (i.e. $\tau_R\in\T$), as Remark~\ref{rem:debut} points out. 
\end{remark}
	
Next, we investigate the measurability of $x\mapsto J(x,R)$. The key in the proof below is the use of a non-standard measurable projection theorem (i.e., Theorem~\ref{thm:projection} below) that does not require Borel structure, as is needed in all classical projection theorems. Section \ref{sec:proj} will be devoted to the derivation of this new measurable projection theorem.

\begin{lemma}\label{measuretheta}
Suppose $\{(x,\cP(x)): x\in\R^d\}\subseteq \R^d\times \fP(\Omega)$ is universally measurable and $g:\R^d\to\R$ is universally measurable.
For any $R\in\cU(\R^d)$, the functions
\begin{equation}\label{maps}
x \mapsto \inf_{\P\in\cP(x)} \E^{\P}[e^{-r\rho_R}g(B_{\rho_R})]\quad \hbox{and}\quad x \mapsto \sup_{\P\in\cP(x)} \E^{\P}[e^{-r\rho_R}g(B_{\rho_R})]
\end{equation}
are universally measurable. Hence, $x\mapsto J(x,R)$ in \eqref{J} is universally measurable.
%Hence, $\Theta(R)\in \cU(\R)$ and $\Theta$ can be taken as an operator $\Theta: \cU(\R)\mapsto \cU(\R)$.
\end{lemma}

\begin{proof}
By Lemma~\ref{lem:in T}, universal measurability of $g$, and Proposition 7.44 of \cite{BS-book-78}, the function $e^{-r\rho_R}g(B_{\rho_R})$, mapping $\Omega$ to $\R$, is universally measurable. It follows that $f(\P):=\E^{\P}[e^{-r\rho_R}g(B_{\rho_R})]$, viewed as a map from $\fP(\Omega)$ to $\R$, is universally measurable, thanks to Corollary 7.46.1 of \cite{BS-book-78}. For any $K\in \R$, this implies that 
\begin{equation}\label{the set}
\left\{\P\in \fP(\Omega): \E^{\P}[e^{-r\rho_R}g(B_{\rho_R})]<K \right\}
\end{equation} 
is universally measurable. Consequently,
\begin{equation}\label{A}
A := \big\{(x,\cP(x)):x\in\R^d\big\} \cap \left(\R^d\times\left\{\P\in \fP(\Omega): \E^{\P}[e^{-r\rho_R}g(B_{\rho_R})]<K  \right\}\right) 
\end{equation}
is universally measurable, thanks to the assumption that $\left\{(x,\cP(x)):x\in\R^d\right\}$ is universally measurable. Now, observe that 
\begin{align}\label{projection}
&\left\{x\in\R^d : \inf_{\P\in\cP(x)} \E^{\P}[e^{-r\rho_R}g(B_{\rho_R})]<K\right\} = \text{proj}_{\R^d} (A), %\\
%&\hspace{0.2in} = \text{proj}_{\R^d}\left(\left\{(x,\cP(x)):x\in\R^d\right\} \bigcap \left(\R^d\times\left\{\P\in \fP(\Omega): \E^{\P}[e^{-r\rho_R}g(B_{\rho_R})]<K  \right\}\right)  \right)
\end{align}
which is universally measurable in $\R^d$, thanks to the generalized measurable projection result Theorem~\ref{thm:projection}. % and Remark \ref{measurermk}, a generalized measurable projection result, for which the proof will be presented in the next section.
Thus, we conclude that $x \mapsto \inf_{\P\in\cP(x)} \E^{\P}[e^{-r\rho_R}g(B_{\rho_R})]$ is universally measurable. By a similar argument and the fact that the complement of a universally measurable set is universally measurable, we obtain that $x \mapsto \sup_{\P\in\cP(x)} \E^{\P}[e^{-r\rho_R}g(B_{\rho_R})]$ is also universally measurable. Hence, by definition in \eqref{J}, $x\mapsto J(x,R)$ is universally measurable.
\end{proof}

\begin{remark}\label{clremk}
With $R\in\cU(\R^d)$, classical measurable projection theorems (see e.g. Theorem 2.12 of \cite{Crauel-book-02} or Theorem III.23 of \cite{CV-book-77}) cannot be used in the proof above: %of Lemma \ref{measuretheta}. 
applying these theorems requires Borel measurability of the set in \eqref{the set}, %$\text{proj}_{\R^d} (A)=\left\{\P\in \fP(\Omega): \E^{\P}[e^{-r\rho_R}g(X_{\rho_R})]<K  \right\}$, 
which is only universally measurable in general. % (as $R\in\cU(\R^d)$ need not belong to $\B(\R^d)$). % needs to be a Borel subset of $\fP(\Omega)$. 
Theorem~\ref{thm:projection} then comes into play:  it obtains the same projection result as classical theorems, but does not require Borel measurability to start with. 
%However, even for $R\in\B(\R^d)$, $\rho_R$ in general is only universally measurable; indeed, for any $t\ge 0$, $\{\rho_R\le t\}$ lies in $\F^\P_t$ for all $\P\in\fP(\Omega)$, but not necessarily in $\F^B_t$. Thus, even with Borel measurable $g$, it is only  guaranteed that $\P\mapsto \E^{\P}[e^{-r\rho_R}g(X_{\rho_R})]$ is universally measurable, but not Borel measurable. 
\end{remark}

It can now be established that $\Theta(R)$ remains in $\cU(\R^d)$ whenever $R\in\cU(\R^d)$. 

\begin{proposition}\label{prop:measuretheta}
Suppose $\{(x,\cP(x)): x\in\R^d\}\subseteq \R^d\times \fP(\Omega)$ is universally measurable and $g:\R^d\to\R$ is universally measurable. Then, for any $R\in\cU(\R^d)$, $\Theta(R)\in \cU(\R^d)$. 
\end{proposition}

\begin{proof}
As $x\mapsto J(x,R)$ is universally measurable (Lemma~\ref{measuretheta}) and $g$ is also universally measurable, $S_R$, $I_R$, and $C_R$, defined in \eqref{regions}, all belong to $\cU(\R^d)$. Hence, $\Theta(R) = S_R\cup (I_R\cap R) \in \cU(\R^d)$.
\end{proof}

\subsubsection{Discussion on the use of universally measurable stopping policies}\label{why use cU}
Our choice to work with stopping policies $R\in\cU(\R^d)$, instead of $R\in\B(\R^d)$, can now be explained clearly. In view of Proposition~\ref{prop:measuretheta}, $\Theta$ defined in \eqref{opertheta} is an operator acting on $\cU(\R^d)$, i.e., $\Theta:\cU(\R^d)\to \cU(\R^d)$. This is a desirable property: it facilitates the definition of equilibria in Definition~\ref{def:E} below, as well as the fixed-point iteration in \eqref{iteration}. 

If we instead work with $R\in\B(\R^d)$, it is not necessarily true that $\Theta(R)\in\B(\R^d)$, i.e., $\B(\R^d)$ is not closed under the operator $\Theta$. 
Specifically, when we focus on $R\in\B(\R^d)$, Lemma~\ref{measuretheta} can be revised as follows. Suppose that $\{(x,\cP(x)): x\in\R^d\}\subseteq \R^d\times \fP(\Omega)$ is Borel measurable and $g:\R^d\to\R$ is Borel measurable. As $\tau_R$ is now a Borel measurable random variable (see Remark~\ref{rem:two Borel concepts}), we argue as in the proof of Lemma~\ref{measuretheta} that $f(\P):=\E^{\P}[e^{-r\rho_R}g(B_{\rho_R})]$ is Borel measurable, thanks to Corollary 7.29.1 of \cite{BS-book-78}. It follows that the set in \eqref{the set} is now Borel measurable, and so is the set $A$ in \eqref{A}. With such enhanced measurability of $A$ (from universal to Borel), we can apply classical measurable projection theorems (see e.g. Theorem 2.12 of \cite{Crauel-book-02} or Theorem III.23 of \cite{CV-book-77}) to conclude that $\text{proj}_{\R^d} (A)$ is universally measurable, with no need of the generalized projection result Theorem~\ref{thm:projection}. Note, however, that $\text{proj}_{\R^d} (A)$ is only universally measurable in general even when $A$ is now Borel. This reflects the general fact that the projection of a Borel measurable set need not be Borel measurable; see Section 7.6 of \cite{BS-book-78}. With only universal measurability of $\text{proj}_{\R^d} (A)$ in \eqref{projection} (and thus of the maps in \eqref{maps} and $x\mapsto J(x,R)$), the sets $S_R$, $I_R$, $C_R$, and therefore $\Theta(R)$, are only universally measurable in general. To wit, with $R\in\B(\R^d)$, $\Theta(R)\in\B(\R^d)$ is not guaranteed. 

\begin{remark}\label{rem:MA needs cU}
The need of $R\in\cU(\R^d)$, as detailed above, is caused by model ambiguity. Indeed, in the absence of ambiguity (i.e. for all $x\in\R^d$, $\cP(x)=\{\P^x\}$ for some $\P^x\in\fP(\Omega)$), \eqref{J} reduces to $J(x,R) =\E^{\P^x}[e^{-r\rho_R}g(B_{\rho_R})]$,  where $B$ is strong Markov under $(\P^x)_{x\in\R^d}$; recall \eqref{cP}. As long as $g$ is Borel measurable, $x\mapsto  J(x,R)$ is Borel measurable for all $R\in\B(\R^d)$. %There is then no need to consider more general $R\in\cU(\R^d)$. 
With measurability much easier to achieve now (with no need of any projection), one can simply focus on $R\in\B(\R^d)$. 
\end{remark}

%\begin{remark}\label{why use cU}
%%Taking stopping policies $R$ from $\cU(\R^d)$, instead of $\B(\R^d)$ or the collection of analytic sets in $\R^d$, ensures that $\Theta$ is an operator, which facilitates our fixed-point formulation. 
%If we take $R$ to be a Borel (or analytic) set, it is not guaranteed that $\Theta(R)$ is also Borel (or analytic). Indeed, for any $R\in\B(\R^d)$, as $x\mapsto J(x,R)$ is only universally measurable (recall Remark~\ref{clremk}), $\Theta(R)\in \B(\R^d)$ need not hold. Alternatively, take $R$ to be an analytic set in $\R^d$. Even if we show that $x\mapsto J(x,R)$ is upper (resp. lower) semi-analytic and thus $C_R$ (resp. $S_R$) is analytic, it is unclear if $S_R$ (resp. $C_R$) is analytic. Whether $\Theta(R)$ is analytic is then in question.   
%\end{remark}

%%%%%%%%%%%%%%%%%%%%%%%%%%%%

\subsection{Problem Formulation}
In view of Proposition~\ref{prop:measuretheta}, $\Theta$ defined in \eqref{opertheta} can be viewed as an operator acting on $\cU(\R^d)$, i.e., $\Theta:\cU(\R^d)\to \cU(\R^d)$. An equilibrium is then defined as a fixed point of the operator.

\begin{definition}\label{def:E}
$R\in\cU(\R^d)$ is called an equilibrium if $\Theta(R)=R$. We denote by $\cE$ the collection of all equilibria.
\end{definition}

\begin{remark}[Existence of an equilibrium]
The entire space $\R^d$ is an equilibrium. Indeed, for any $x\in\R^d$, $\rho_{\R^d} = 0$ and thus $J(x,\R^d)=g(x)$. This implies $I_{\R^d}=\R^d$, so that $\Theta(\R^d) =\R^d$. 
\end{remark}

The general methodology for finding equilibria (other than the entire space $\R^d$) is to perform fixed-point iterations: one starts with an arbitrary $R\in\cU(\R^d)$, and apply $\Theta$ to it repetitively until an equilibrium is reached. That is, we take 
\begin{equation}\label{iteration}
R_* := \lim_{n\to\infty} \Theta^n(R)
\end{equation}
as a candidate equilibrium. 
%{\color{blue}It is noted that the limit of sets in \eqref{iteration} should be interpreted in either of two equivalent ways: (i) the upper and lower bounds on the sequence that converge monotonically to the same set and (ii) by convergence of a sequence of indicator functions which are themselves real-valued. In some special cases such as when $\Theta^n(R)$ is a monotone sequence of sets, the convergence in \eqref{iteration} holds trivially.}
To make this fixed-point approach rigorous, two important questions need to be answered: 
\begin{itemize}
\item [(i)] How do we make sense of the limit-taking in \eqref{iteration}, and further show that $R_*$ is well-defined?
\vspace{-0.15in}\\
The right hand side of \eqref{iteration} involves a sequence of sets in $\R^d$, whose convergence can be complicated and there is no standard definition of it. Nonetheless, when the state process is a one-dimensional diffusion, Section~\ref{sec:converg} below shows that $(\Theta^n(R))_{n\in\N}$ is nondecreasing, so that $R_*$ can simply be defined as $\bigcup_{n\in\N}\Theta^n(R)$; see Proposition~\ref{limitstopping}.
 \item [(ii)] Given that $R_*$ is well-defined, is it indeed an equilibrium (i.e., $\Theta(R_*)=R_*$)?
\end{itemize}
%{
%\color{blue}Up to this point, we are in a general multi-dimensional framework. It has been shown that the iteration operator $\Theta(R)$ in \eqref{opertheta} is always well defined thanks to our new generalized measurable projection theorem in Section \eqref{sec:proj}.} To make our iteration argument work, we need to show rigorously that: (i) the above limit-taking is well-defined for the sequence of sets $\Theta^n(R)$ and we have $R_*\in\cU(\R^d)$, and (ii) $R_*$ is indeed an equilibrium, i.e., $\Theta(R_*)=R_*$. {\color{blue}In general,  such results are very challenging to establish. However, we hope our general set up and conclusions can help to pave the way for some future research in a general multi-dimensional framework.}
The next section focuses on answering (i) and (ii) in a concrete one-dimensional diffusion model under drift and volatility uncertainty. %parameter ambiguity
% such that $\{\Theta^n(R)\}_n$ can be shown as a nondecreasing sequence. % in a general formulation of model ambiguity. 

%%%%%%%%%%%%%%%%%%%%%%%%%%%%%%%%%%%%%%
%%%%%%%%%%%%%%%%%%%%%%%%%%%%%%%%%%%%%%

\section{Convergence of Fixed-Point Iterations}\label{sec:converg}
In this section, under a strong formulation of model ambiguity, we will show that the fixed-point iteration \eqref{iteration} indeed converges to an equilibrium, when the involved state process is a one-dimensional diffusion. Our analysis crucially relies on the ``regular'' property of one-dimensional diffusion processes (i.e., \eqref{regular} below). The multi-dimensional case is left for future research.%{\color{blue}which requires some distinctive mathematical methods.}
\footnote{For a multi-dimensional diffusion process, there is no corresponding notion of being ``regular''. As a result, the analysis in this section does not extend naturally to a multi-dimensional case.}

%We have confirmed that $\Theta: \cU(\R)\mapsto \cU(\R)$ is well defined given the assumption that $\{(x,\cP(x)): x\in\R\}$ is universally measurable. For specific models in this section, we need to first verify that  $\{(x,\cP(x)): x\in\R\}$ is indeed universally measurable and then show that the iteration converges so that the fixed point argument will conclude the existence of the equilibrium stopping policy. However, as we start with a family of probabilities in $\cP(x)$, the convergence is expected to hold in the quasi sure sense and the proof becomes highly technical in a general setting. It is the scope of this paper to only consider one dimensional diffusion process with model uncertainty in the strong formulation, so that Markov properties and quasi sure analysis can become helpful ingredients in our proofs. The more general setting with model uncertainty, especially the multi-dimensional case with correlation uncertainty will be left for our future research.
    	
Take $d=1$ in the set-up of Section~\ref{sec:model}. Let $\P_0\in\fP(\Omega)$ denote the Wiener measure, under which $B$ is a standard Brownian motion. %under $\P_0$, In \eqref{SDE}, $\int_{0}^{t} \sigma(X^{x, b,\sigma}_t) dB_t$ is the standard stochastic integral with respect to Brownian motion. 
Let $I=(\ell,r)$, for some $-\infty\le \ell <r \le \infty$, be a given interval. For any $x\in{I}$, consider the stochastic differential equation
\begin{equation}\label{SDE}
\quad X^{x, b,\sigma}_t = x + \int_0^t b(X^{x, b,\sigma}_s) ds +  {\int_{0}^{t}} \sigma(X^{x, b,\sigma}_s) dB_s,\quad \forall 0\le t<\zeta,\qquad \P_0\hbox{-a.s.},
\end{equation}
where $\zeta:= \lim_{n\to\infty} S_n$, with $S_n := \inf\{t>0:X^{x,b,\sigma}_t\notin (\ell+1/n,r-1/n)\}$ for all $n\in\N$. 
We assume that $X^{x,b,\sigma}$ is absorbed at the endpoints of $I$ in the case $\zeta<\infty$. For any $y\in I$ and $A\in\cU(I)$, we introduce the hitting times
\begin{equation}\label{hitX}
T^{x,b,\sigma}_y := \inf\{t> 0: X^{x,b,\sigma}_t =y\}\qquad \hbox{and}\qquad T^{x,b,\sigma}_A := \inf\{t> 0: X^{x,b,\sigma}_t \in A\}.
\end{equation}
For simplicity, we will often write $X$, $T^x_y$, and $T^x_A$ for $X^{x,b,\sigma}$, $T^{x,b,\sigma}_y$, and $T^{x,b,\sigma}_A$.

%%%%%%%%%%%%%%%%%%%%%%

\subsection{A Strong Formulation of Model Ambiguity}
Let us introduce a subset of $\{(b,\sigma): b,\sigma\ \hbox{maps $I$ to $\R$}\}$ that help specify the scope of uncertainty we would like to deal with. 

\begin{definition}\label{def:Gamma}
Let $\mathfrak{L}$ be the collection of functions $b, \sigma:I\to \R$ that are Lipschitz continuous and grows at most linearly on $I$, with $\sigma^2(y)>0$ for all $y\in I$. Moreover,
\begin{itemize}
\item [(i)] let $\A$ be the collection of all set-valued functions $\Pi: I\to 2^{\mathfrak L}$;
\item [(ii)] let $\A^\infty$ be the collection of all set-valued functions $\Pi: I\to 2^{\mathfrak L}$ satisfying the following: for any $x\in I$, there exists $K>0$ such that for any $(b,\sigma)\in\Pi(x)$,
\[
|b(u)-b(v)|+|\sigma(u)-\sigma(v)|\le K|u-v|\quad \hbox{and}\quad |b(u)|+|\sigma(u)|\le K(1+|u|),\quad \forall u, v\in I.
\]
\end{itemize}
% such that $\Pi(x)\subseteq\Gamma$ for all $x\in I$.
%satisfying (i) $\sigma^2(y)>0$ for all $y\in I$, and (ii) $b$ and $\sigma$ are Lipschitz and grows at most linearly.  
%\begin{itemize}
%\item [(i)] for any $y\in I$, $\sigma^2(y)>0$ and $\int_{y-\eps}^{y+\eps}\frac{|b(z)|}{\sigma^2(z)}dz<\infty$ for some $\eps>0$; %is locally integrable on $I$,%for all $y\in I$
%\item [(ii)] \eqref{SDE} admits a unique strong solution. 
%\end{itemize}
%For any set-valued map $\Pi: I\to\{(b,\sigma): b,\sigma\ \hbox{maps $I$ to $\R$}\}$ such that $\Pi(x)\subseteq\Gamma(x)$ for all $x\in I$, we will write $\Pi\subseteq \Gamma$ for brevity. 
\end{definition}
Here, %$\mathfrak L$ represents the largest extent of ambiguity, independent of the current state $x\in I$, our framework can accommodate. On the other hand, 
each $\Pi:I\to 2^{\mathfrak L}$ identifies the actual ambiguity faced by the agent, depending on the current state $x\in I$. That is, $\Pi(x)\subseteq\mathfrak L$ is the collection of coefficients $(b,\sigma)$ in \eqref{SDE} that are considered plausible by the agent at $x\in I$. 
%practical situations and the agent's subjective beliefs. 

%Note that Definition~\ref{def:Lambda} (i) is by no means restrictive: it is a fairly general condition related to the existence of a weak solution to \eqref{SDE} that is unique in distribution; see Theorem 5.5.15 in \cite{KS-book-91}. In many applications where a unique strong solution exists (as required in Definition~\ref{def:Lambda} (ii)), Definition~\ref{def:Lambda} (i) can be easily verified. More importantly, Definition~\ref{def:Lambda} (i) gives rise to %Definition~\ref{def:Lambda} (i) has a desired consequence: the {\it regularity} of $X^{x,b,\sigma}$; see Section~\ref{subsec:regular} for details. 

For each $x\in{I}$ and $(b,\sigma)\in\mathfrak L$, the Lipschitz and linear growth conditions in Definition~\ref{def:Gamma} ensure the existence of a unique strong solution $X^{x,b,\sigma}$ to \eqref{SDE}. By viewing $X^{x,b,\sigma}$ as a map from $\Omega$ to itself, we define the  probability measure $\P^x_{b,\sigma}\in\fP(\Omega)$ by
\begin{equation}\label{P^x}
\P^x_{b,\sigma} := \P_0 \circ (X^{x,b,\sigma})^{-1}.
\end{equation}
By construction, for any $A\subseteq \cU(\Omega)$,
\begin{equation}\label{P^x to P_0}
\P^x_{b,\sigma}(A) = \P_0\left(\{\omega\in\Omega: X^{x,b,\sigma}(\omega)\in A\}\right). 
\end{equation}
%\begin{definition}
%Now, for any $x\in I$, the collection of probabilities associated with $\mathfrak L$ is given by
%\[
%\cP_{\mathfrak L}(x):= \{\P^x_{b,\sigma} : (b,\sigma)\in\mathfrak L\}.
%\]
Given $\Pi\in\A$, we introduce
\begin{equation}\label{cP Pi}
\cP(x) :=  \{\P^x_{b,\sigma} : (b,\sigma)\in\Pi(x)\},\quad \forall x\in I.
\end{equation}
%\end{definition}

\begin{remark}\label{rem:rectangular}
The ``rectangularity'' of the set of priors (cf. \cite{ES03}) was refined in \cite{NvH13} as ``closedness under conditioning'' plus ``stability under pasting'' (i.e., Assumption 2.1 (ii) and (iii) therein). While ``rectangularity'' is widely assumed in the literature of model ambiguity (recall discussions below \eqref{consistent condition}), we do not impose it on $\{\cP(x)\}_{x\in I}$. 

To see this, note that ``closedness under conditioning'' in our setting amounts to the following: %it is necessary that each $\cP(x)$ contains conditional probabilities of all probabilities in $\cP(y)$, for all $y\neq x$. Specifically, 
given $x\in I$ and $(b,\sigma)\in\Pi(x)$, for any $t>0$, $(\P^x_{b,\sigma}\mid\F_t) (\omega)\in \cP(X_t(\omega))$ for $\P^x_{b,\sigma}$-a.e. $\omega\in\Omega$, 
where $\P^x_{b,\sigma}\mid\F_t$ denotes the conditional probability of $\P^x_{b,\sigma}$ given $\F_t$.  As $X^{x,b,\sigma}$ is a time-homogeneous Markov process, we have $(\P^x_{b,\sigma}\mid\F_t)(\omega)=\P^{X_t(\omega)}_{b,\sigma}$, so that the above condition becomes 
\begin{equation}\label{closedness}
\P^{X_t(\omega)}_{b,\sigma} \in \cP(X_t(\omega))\quad \hbox{for $\P^x_{b,\sigma}$-a.e. $\omega\in\Omega$}. 
\end{equation}
This can be easily violated by $\{\cP(x)\}_{x\in I}$ in \eqref{cP Pi}. Specifically, take $A\subseteq I$ with $\P^x_{b,\sigma}(X_t \in A)>0$. For any $\Pi: I\to 2^{\mathfrak L}$ in Definition~\ref{def:Gamma}, as long as $(b,\sigma)\notin \Pi(y)$ for all $y\in A$, we have $\P^y_{b,\sigma}\notin \cP(y)$ for all $y\in A$ and thus \eqref{closedness} is violated on the set $\{X_t\in A\}$. 
%$\P^x_{b,\sigma}\otimes_t \Q \in \cP(x)$ for all $\Q\in\cP(\omega_t)$, i.e. 
%\begin{equation}\label{pasting}
%\P^x_{b,\sigma}\otimes_t \P^{\omega_t}_{b',\sigma'} \in \cP(x)\quad \hbox{for all}\ (b',\sigma')\in\Pi(\omega_t). 
%\end{equation}
%As long as there is $(b',\sigma')\in\Pi(\omega_t)$ that does not coincide with $(b,\sigma)\in\Pi(x)$ on $\{\omega_s:0\le s\le t\}\subseteq I$, the space travelled by $\omega$ up to time $t$, we have $\P^x_{b,\sigma}\otimes_t \P^{\omega_t}_{b',\sigma'}\notin\cP(x)$ and thus \eqref{pasting} fails. %thanks to \eqref{cP Pi} and \eqref{P^x}. 
%However, by \eqref{cP Pi} and \eqref{P^x}, the only possibility for $\P^x_{b,\sigma}\otimes_t \Q \in \cP(x)$ is $\Q= \P^{\omega_t}_{b,\sigma}$.
%That is, \eqref{pasting} requires $\cP(y)$ to contain exactly the measure $\P^y_{b,\sigma}$ for all $y\in I$, reducing to the case with no ambiguity where $(b,\sigma)$ is fixed throughout. 
%As $X^{x,b,x\sigma}$ is a regular diffusion.  That is, we need $\P^{\omega_t}_{b,\sigma}\in \cP(\omega_t)$. 
\end{remark}

Two important consequences of Definition~\ref{def:Gamma}, the {\it regularity} of $X^{x,b,\sigma}$ and the {\it relative compactness} of $\cP(x)$, are established in Lemmas \ref{lem:regular} and ~\ref{lem:rc} below. 

\begin{lemma}\label{lem:regular}
For any $x\in {I}$ and $(b,\sigma)\in\mathfrak L$, $X^{x,b,\sigma}$ is a {\it regular} diffusion, i.e.,  
\begin{equation}\label{regular}
\hbox{for any}\ x\in {I},\quad \P_0(T^{x}_y <\infty) >0,\quad \forall y\in I.
\end{equation}
\end{lemma}

\begin{proof}
Under Definition~\ref{def:Gamma}, the scale function
\[
s(z) := \int_x^z \exp\left(-2\int_x^u\frac{b(\xi)}{\sigma^2(\xi)}d\xi\right) du,\quad z\in I,
\]
is well-defined, strictly increasing, and continuously differentiable. Let $q:(s(\ell), s(r))\to \R$ be the inverse function of $s$. By the arguments in Proposition 5.5.13 of \cite{KS-book-91}, $X$ being the unique strong solution to \eqref{SDE} entails the existence of a unique strong solution to 
$dY_t = \tilde\sigma(Y_t) dB_t$, $Y_0=0$, $\P_0\hbox{-a.s.}$, where 
$
\tilde\sigma(y) := s'(q(y)) \sigma(q(y))$ for $s(\ell)<y< s(r). 
$

By Theorem 5.5.4 in \cite{KS-book-91} and $\sigma^2>0$ on $I$, $\tilde\sigma^2$ is locally integrable. Hence, the speed measure
$
m(dy) := \frac{2 dy}{\tilde\sigma^2(y)}$,  $s(\ell)<y< s(r)$,
assigns a finite value to any $[a,b]\subset (s(\ell),s(r))$. This readily implies that $Y$ is a regular diffusion; see Remark (ii), ``{\it The converse to Theorem 47.1}'', on p.277 of \cite{RW-book-00}. As $X=s^{-1}(Y)$, $X$ is also regular. 
\end{proof}

Te fact that $X$ is a regular diffusion (i.e., satisfying \eqref{regular}) means that $I=(\ell,r)$ cannot be decomposed into smaller intervals from which $X$ could not exit. Moreover, when starting with $x\in I$, $X$ has to enter the regions above and below $x$ immediately, as stated below. 

\begin{remark}\label{lem:T^x_x=0}
For any $x\in {I}$ and $(b,\sigma)\in\mathfrak L$, $T^x_{(\ell, x)}=T^x_{(x,r)} = T^x_x= 0$ $\P_0$-a.s. Indeed,
%\begin{proof}
Lemma 46.1 (i) in \cite{RW-book-00} directly gives  $T^x_{(\ell, x)}=T^x_{(x,r)}=0$ $\P_0$-a.s.; see also the discussion above Lemma 46.1 therein. Now, for $\P_0$-a.e. $\omega\in\Omega$, the fact $\P_0(T^x_{(\ell,x)} =0)=\P_0(T^x_{(x,r)} =0)=1$ implies that for any $n\in\N$, there exist $t,t'\in [0,1/n]$ such that $X_t(\omega) >x$ and $X_{t'}(\omega)<x$. Hence, $T^x_x(\omega)\le 1/n$ for all $n\in\N$, implying $T^x_x(\omega) =0$.
\end{remark}

%\begin{remark}
%The standard Lipschitz and linear growth conditions on $b, \sigma$ ensures the existence of a unique strong solution $X$ to \eqref{SDE}, for any $x\in \Int(I)$. Additional conditions on $b, \sigma$ on $\R\setminus \Int(I)$ are needed to ensure that $X$ $I$ is
%
%To further ensure that the state space is $I$,
%\end{remark}
%that are Lipschitz and satisfy the linear growth condition, i.e. there exists $K>0$ such that for any $x,y \in \R$,
%\begin{align*}
%|b(x)-b(y)|+|\sigma(x)-\sigma(y)| \le K|x-y|,\\
%|b(x)|+|\sigma(x)|\le K(1+|x|).
%\end{align*} 

%\begin{remark}
%The distribution of the canonical process $B$ under $\P^{x}_{b,\sigma}$ is the same as that of the unique strong solution $X^{x,b,\sigma}$ to \eqref{SDE} under $\P_0$.
%\end{remark}

\begin{corollary}\label{coro:rho_x=0}
For any $x\in {I}$ and $(b,\sigma)\in\mathfrak L$, $\rho_{(\ell, x)}=\rho_{(x,r)} = \rho_{\{x\}}= 0$ $\P^x_{b,\sigma}$-a.s.
\end{corollary}

\begin{proof}
Observe from \eqref{P^x to P_0} that 
\begin{align*}
\P^x_{b,\sigma}(\rho_{(\ell, x)}=0) &= \P^x_{b,\sigma}\left(\inf\{t>0 : B_t\in (\ell, x)\}=0\right)\\
&= \P_0\left(\inf\{t>0 : X^{x,b,\sigma}_t\in (\ell, x)\}=0\right) = \P_0(T^x_{(\ell, x)}=0) =1,
\end{align*}
where the last equality follows from Remark~\ref{lem:T^x_x=0}. The same argument shows that $\P^x_{b,\sigma}(\rho_{(x,r)}=0)=\P_0(T^x_{(x,r)}=0)=1$ and $\P^x_{b,\sigma}(\rho_{\{x\}}=0)=\P_0(T^x_{x}=0)=1$.
\end{proof}

The following observation will be useful in Section~\ref{sec:example}.

\begin{remark}\label{rem:closed E}
By Remark \ref{lem:T^x_x=0} (or Corollary \ref{coro:rho_x=0}), we can follow arguments in Lemmas 4.1 and 4.2 in \cite{HZ17-continuous} to show that for any $R \in\cU({I})$, 
%\begin{align*}
$\rho_R=\rho_{\overline{R}}\ \mathbb{P}^x_{b,\sigma}\hbox{-a.s.},\ \forall x\in I\ \hbox{and}\ (b,\sigma)\in\mathfrak L.$
%\end{align*}
%In particular, we have $\rho_R=\rho_{\overline{R}}=0$,  $\mathbb{P}^x_{b,\sigma}$-a.s., $\forall (b,\sigma)\in\Lambda(x)$, $\forall x\in \overline{R}$. 
Consequently, $S_R=S_{\overline{R}}$, $I_R=I_{\overline{R}}$, and $C_R=C_{\overline{R}}$. It follows that $R\in\cE$ if and only if $\overline{R}\in\cE$.

This, however, does not imply that we can focus on solely $R\in\B(I)$ in our one-dimensional setting. In the main result Theorem~\ref{thm:main} below, $\Theta(R)$ and $R$ need to have the same measurability to facilitate the fixed-point iteration. As explained in Section~\ref{why use cU}, $R\in\B(I)$ does not guarantee $\Theta(R)\in\B(I)$: this loss of Borel measurability stems from the use of projections in Lemma~\ref{measuretheta}, indispensable under model ambiguity (even when $d=1$). Namely, the use of $R\in\cU(I)$ is caused by model ambiguity, regardless of the dimension of the state space; see also Remark~\ref{rem:MA needs cU}.
\end{remark}

Focusing on $\Pi\in\A^\infty$ yields the relative compactness of $\cP(x)$. 

\begin{lemma}\label{lem:rc}
For any $\Pi\in\A^\infty$, $\cP(x)$ is relatively compact for all $x\in I$. 
\end{lemma} 

\begin{proof}
Fix $x\in I$. By Theorem 1.3.1 of \cite{Stroock-Varadhan-book-06}, $\cP(x)$ is relatively compact if and only if for any $\eps>0$ and $T>0$,
\begin{equation}\label{rc}
\lim_{\delta\downarrow 0}\sup_{\P\in\cP(x)} \P\bigg(\sup_{0\le s\le t\le T,\ t-s<\delta} |B_t-B_s|> \eps\bigg) =0. 
\end{equation}
Thanks to the Lipschitz and linear growth conditions, under the same constant $K>0$, in Definition~\ref{def:Gamma} (ii), standard estimations, see e.g. Proposition 1.2.1 in \cite{Bouchard-note-07}, show that there exist constants $\beta, \gamma>0$ such that for any $(b,\sigma)\in\Pi(x)$, 
\begin{equation}\label{estimate}
\E^{\P_0}[|X^{x,b,\sigma}_t-X^{x,b,\sigma}_s|^\beta]\le C_T|t-s|^{1+\gamma},\quad \forall T>0\ \hbox{and}\ 0\le s,t\le T,
\end{equation}
where $C_T>0$ depends on only $x\in I$, $T>0$, and $K>0$. In view of the proof of Theorem I.2.1 in \cite{RY-book-99}, \eqref{estimate} implies that for any $\eta\in [0,\gamma)$, there exists $C_\eta>0$ such that 
\[
\E^{\P_0}\bigg[\sup_{0\le s\le t\le T} \frac{|X^{x,b,\sigma}_t-X^{x,b,\sigma}_s|^\beta}{|t-s|^\eta}\bigg] \le C_\eta,\quad \forall (b,\sigma)\in\Pi(x). 
\]
This, together with the Markov inequality, shows that for any $\P\in\cP(x)$, 
\begin{align*}
\P\bigg(\sup_{0\le s\le t\le T,\ t-s<\delta} |B_t-B_s|> \eps\bigg)&\le \eps^{-\beta} \E^\P\bigg[\sup_{0\le s\le t\le T,\ t-s<\delta} |B_t-B_s|^\beta\bigg]\\
&=\eps^{-\beta} \E^{\P_0}\bigg[\sup_{0\le s\le t\le T,\ t-s<\delta} |X^{x,b,\sigma}_t-X^{x,b,\sigma}_s|^\beta\bigg]\le \eps^{-\beta} C_\eta \delta^\eta,
\end{align*}
which readily yields \eqref{rc}.
\end{proof}

%%%%%%%%%%%%%%%%%%%%

\subsection{The Main Result}

Corollary~\ref{coro:rho_x=0} facilitates the convergence of the fixed-point iteration \eqref{iteration}, as the next result shows. Recall that for any $\Pi\in \A$, $\cP(x)$ is defined as in \eqref{cP Pi}.

\begin{proposition}\label{limitstopping}
Fix $\Pi\in\A$ such that $\{(x,\cP(x)): x\in {I}\}\subseteq I\times\fP(\Omega)$ is universally measurable. Then, for any $R\in\cU({I})$, $R\subseteq \Theta(R)$. Hence, $R_*$ in \eqref{iteration} is well-defined, and of the form
\begin{equation}\label{union form}
R_* %:= \lim_{n\to\infty} \Theta^n(R) 
= \bigcup_{n\in\N} \Theta^n(R) \in \cU({I}).
\end{equation}
\end{proposition}

\begin{remark}
Assuming universal measurability of $\{(x,\cP(x)): x\in {I}\}$ in $I\times\fP(\Omega)$ is not restrictive in terms of the related literature. Such a set is typically assumed to be analytic (and thus universally measurable) in the more general path-dependent setting; see e.g., \cite{NeufeldNutz13}, \cite{NvH13}, and \cite{Nutz17-MF}.
%Note that by the argument of Theorem 2.4 in \cite{NeufeldNutz13} (or Theorem 2.5 in \cite{NeufeldNutz14}), $\{(x,\cP_\Gamma): x\in {I}\}$ is analytic (and hence universally measurable).
%Assuming universal measurability of $\{(x,\cP(x)): x\in {I}\}$ in $I\times\fP(\Omega)$ is not very restrictive. For instance, for any Borel measurable $E\subseteq\R$ and $\Sigma\subseteq (0,\infty)$, consider the set-valued map $\Pi_{E,\Sigma}: I\to\{(b,\sigma): b,\sigma\ \hbox{maps $I$ to $\R$}\}$ be such that $\Pi_{E,\Sigma}\subseteq\Gamma$ and for any $x\in I$, $\Pi_{E,\Sigma}(x)$ is the collection of all $(b,\sigma)$ satisfying $b(I)\subseteq E$ and $\sigma^2(I)\subseteq\Sigma$. 
%In particular, with $E=\R$ and $\Sigma=(0,\infty)$,  we have $\Pi_{E,\Sigma}=\Gamma$.
%Following the argument of Theorem 2.4 in \cite{NeufeldNutz13} or Theorem 2.5 in \cite{NeufeldNutz14}, $\cP(x)$, induced by $\Pi_{E,\Sigma}$, is Borel measurable and $\{(x,\cP(x)): x\in {I}\}$ is analytic (and hence universally measurable). 
\end{remark}

\begin{proof}
Fix $R\in\cU({I})$. For any $x\in R$, we claim that $\rho_R=0$ $\P^x_{b,\sigma}$-a.s. for all $(b,\sigma)\in\Pi(x)$. There are three cases: (i) if $x$ is an interior point of $R$, the claim trivially holds; (ii) if $x$ is a boundary point of $R$, note that $\rho_{(\ell, x)}=\rho_{(x,r)} = 0$ $\P^x_{b,\sigma}$-a.s. (by Corollary~\ref{coro:rho_x=0}) readily implies $\rho_R=0$ $\P^x_{b,\sigma}$-a.s., for all $(b,\sigma)\in\Pi(x)$; (iii) if $x$ is an isolated point of $R$, note that $\rho_{\{x\}}= 0$ $\P^x_{b,\sigma}$-a.s. (by Corollary~\ref{coro:rho_x=0}) and the fact $x\in R$ readily imply $\rho_R=0$ $\P^x_{b,\sigma}$-a.s., for all $(b,\sigma)\in\Pi(x)$. With $\rho_R=0$ $\P^x_{b,\sigma}$-a.s. for all $(b,\sigma)\in\Pi(x)$, we have $J(x,R) = g(x)$, i.e. $x\in I_R$. Hence, we conclude $R\subseteq I_R$, which gives $\Theta(R) = S_R\cup (I_R\cap R) = S_R \cup R \supseteq R$.
This, together with Proposition~\ref{prop:measuretheta}, shows that $\{\Theta^n(R)\}_{n\in\N}$ is an nondecreasing sequence of sets in $\cU(I)$, leading to the last assertion. 
\end{proof}

%\begin{lemma}
%Let $\cP(x) := \{\P^x_{b,\sigma} : (b,\sigma)\in\Lambda(x)\}$, for all $x\in\Int{I}$. We have that $\{(x,\cP(x)): x\in\Int{I}\}$ is an analytic set. Therefore $\{(x,\cP(x)): x\in\Int{I}\}$ is universally measurable and $\Theta: \cU(\R)\mapsto \cU(\R)$ is well defined.
%\end{lemma}
%\begin{proof}
%\ \\
%\ \\
%\ \\
%\end{proof}

%\begin{definition}
%$R\in\cU(\Int{I})$ is called an equilibrium stopping policy if $\Theta(R)=R$, i.e., the stopping policy can not improve anymore. Let us denote $\mathcal{E}(\Int{I})$ as the set of all equilibria.
%\end{definition}

The remaining question is whether the limit $R_*$ of the fixed-point iteration \eqref{iteration} is indeed an equilibrium. To answer this, we need the following technical result, which requires $\Pi\in \A^\infty$. 

%not only converges, but converges to an equilibrium. To state the result, for any $R\in\cU(I)$ and its fixed-point limit $R_0= \bigcup_{n\in\N} \Theta^n(R)$, we will simply write $\rho^n$ and $\rho^0$ for the hitting times $\rho_{R_n}$, with $R_n:= \Theta^{n-1}(R)$, and $\rho_{R_0}$ (recall \eqref{rhoR}), respectively.

\begin{lemma}\label{lem:in capacity}
For any nondecreasing sequence $\{R_n\}_{n\in\N}$ in $\cU(I)$, set $R_0 := \bigcup_{n\in\N} R_n$ and let $\rho^n$ and $\rho^0$ denote the hitting times $\rho_{R_n}$ and $\rho_{R_0}$, defined in \eqref{rhoR}, respectively. Then, for any $x\in I$, 
\begin{equation}\label{rho converge}
\rho^n(\omega) \downarrow \rho^0(\omega),\quad \forall \omega\in \Omega^x. %\ \hbox{with}\ \omega_0\in I.
\end{equation}
Furthermore, for any $\Pi\in\A^\infty$ and $\eps>0$, we have
\begin{align}
&\lim_{n\to\infty}\sup_{\P\in\cP(x)}\P\left(\left|\rho^n-\rho^0\right|\geq \eps\right)=0,\label{capty1} \\
&\lim_{n\to\infty}\sup_{\P\in\cP(x)}\P\left(\left|B_{\rho^n}-B_{\rho^0}\right| 1_{\{\rho^n<\infty\}}\geq \eps\right)=0.\label{capty2}
\end{align}
%where $\rho^n$ and $\rho^0$ denote the hitting times $\rho_{R_n}$ and $\rho_{R_0}$, defined in \eqref{rhoR}, respectively.
\end{lemma}

\begin{remark}\label{rem:infty-infty}
In \eqref{capty1}, $\rho^n$ and $\rho^0$ may take the value $\infty$. In particular, on $\{\rho^0=\infty\}$, $\rho^n=\rho^0 =\infty$ and we define $\rho^n-\rho^0=0$, for all $n\in\N$. This is consistent with \eqref{at infty}, where we do not distinguish between any two stopping times when they both take the value $\infty$. 
%Assumption~\ref{standing}, where $e^{-r\tau}|g(B_\tau)|:= \limsup_{t\to\infty} e^{-rt}|g(B_t)|$ whenever $\tau=\infty$.  
\end{remark}

The proof of Lemma~\ref{lem:in capacity}, relying crucially on both the relative compactness of $\cP(x)$ and the regularity of $X^{x,b,\sigma}$, is relegated to Appendix~\ref{sec:proof of lem:in capacity}.

Now, we are ready to present the main result of this paper.

\begin{theorem}\label{thm:main}
Fix $\Pi\in\A^\infty$ such that $\{(x,\cP(x)): x\in {I}\}\subseteq I\times\fP(\Omega)$ is universally measurable. Suppose that $g:\overline I\to \R$ is continuous and 
\begin{align}\label{condition}
\lim_{t\rightarrow\infty}e^{-rt}g(X^{x, b,\sigma}_t)=0\quad \P_0\text{-a.s.},\quad \forall x\in I\ \hbox{and}\ (b,\sigma)\in\Pi(x).
\end{align}
Then, for any $R\in\cU({I})$, $R_*$ defined as in \eqref{union form} belongs to $\cE$. Hence, 
\begin{align}\label{equistopping}
\mathcal{E} = \left\{\lim_{n\to\infty} \Theta^n(R): R\in  \cU({I})\right\}.
%\left\{R_0\in \cU({I}): R_0= \lim_{n\to\infty} \Theta^n(R)\ \text{for some}\ R \in\cU({I})\right\}=\bigcup_{R\in  \cU({I})}\left\{\lim_{n\to\infty} \Theta^n(R)\right\}.
\end{align}
\end{theorem}

\begin{proof}
Fix $R\in\cU(I)$, and consider $R_*$ defined in \eqref{iteration}. Recall from Proposition \ref{limitstopping} that $R_n:= \Theta^n(R)$, $n\in\N\cup\{0\}$, form a nondecreasing sequence in $\cU(I)$ and $R_*= \bigcup_{n\in\N} R_n$. We will denote by $\rho^n$ and $\rho^*$ the hitting times $\rho_{R_n}$ and $\rho_{R_*}$, defined in \eqref{rhoR}, respectively. 

To show $R_*\in\cE$, i.e., $\Theta(R_*)=R_*$, we first note that it suffices to prove $S_{R_*}\subseteq R_*$. This is because $\Theta(R_*)=S_{R_*}\cup R_*$, thanks to the proof of Proposition \ref{limitstopping}. To this end, foy any $x\notin R_*$, we aim to show that $x\notin S_{R_*}$.  As $R_*= \bigcup_{n\in\N} R_n$, we have $x\notin R_n=\Theta^{n}(R)$ for all $n\in\N$. In view of \eqref{opertheta} and \eqref{regions}, this implies 
\begin{align}\label{base}
J(x,R_{n-1})=J(x,\Theta^{n-1}(R))\geq g(x),\quad \forall n\in\N. 
\end{align}
If we can show that 
\begin{equation}\label{to show}
J(x,R_*)\ge \liminf_{n\to\infty} J(x,R_{n}),
\end{equation}
we immediately obtain $J(x,R_*)\ge g(x)$ from \eqref{base}, and thus $x\notin S_{R_*}$, as desired. The rest of the proof focuses on deriving \eqref{to show}. 

First, let us consider
\begin{align*}%\label{p,q}
p:=\sup\{y\in R_*:y<x\}\quad \text{and}\quad  q:=\inf\{y\in R_*: y>x\},
\end{align*}
where we take $p=\ell$ (resp. $q=r$) if there is no $y<x$ (resp. $y>x$) lying in $R_0$. Similarly, we define
\begin{align*}%\label{p,q}
p_n:=\sup\{y\in R_n:y<x\}\quad \text{and}\quad  q_n:=\inf\{y\in R_n: y>x\},\quad \forall n\in\N.
\end{align*}
As $R_*= \bigcup_{n\in\N} R_n$ and $\{R_n\}_{n\in\N}$ is nondecreasing, we have $p_n\uparrow p$ and $q_n\downarrow q$. For the case where $p_n=p$ and $q_n=q$ for $n$ large enough, $\rho^n=\rho^0$ on $\Omega^x$, and thus $J(x,R_n)=J(x,R^*)$, for all $n$ large enough. That is, \eqref{to show} holds trivially. Hence, in the rest of the proof, we assume that $p_n$ is strictly increasing, or $q_n$ is strictly decreasing. 

Take $\eta>0$, and choose $n^*\in\N$ such that $\max\{|p_n-p|,|q_n-q|\}<\eta$ for all $n\ge n^*$. Note that there exists $M>0$ such that
\begin{equation}\label{bdd}
e^{-r\rho^n} |g(B_{\rho^n})|<M,\quad \forall n\ge n^*,\quad \P\hbox{-a.s.},\quad \hbox{for all }\P\in\cP(x).  
\end{equation}
%the set of random variables $\{e^{-r\rho^n} g(B_{\rho^n})\}_{n\ge n^*}$ is $\P$-a.s. bounded, for all $\P\in\cP(x)$. 
Indeed, for any $n\ge n^*$, if $\rho^n=\infty$, $e^{-r\rho^n} g(B_{\rho^n})=0$ $\P$-a.s. thanks to \eqref{condition}; for all $n\ge n^*$ such that $\rho^n<\infty$, as $B_{\rho^n}$ takes values on $([p-\eta,p]\cup[q,q+\eta])\cap I$, the continuity of $g$ yields the desired boundedness. Thus, by the dominated convergence theorem and \eqref{rho converge},
\begin{equation}\label{dominated}
\lim_{n\to\infty}\E^{\P}[e^{-r\rho^n}g(B_{\rho^n})]=\E^{\P}[e^{-r\rho^*}g(B_{\rho^*})],\quad \forall\P\in\cP(x).
\end{equation}
On the other hand, by the definition of $J$ in \eqref{J}, \eqref{base} implies that for any $\P\in\cP(x)$, 
\begin{align*}
\alpha \E^{\P}[e^{-r\rho^n}g(B_{\rho^n})] + (1-\alpha) \sup_{\P\in\cP(x)} \E^{\P}[e^{-r\rho^n}g(B_{\rho^n})] \ge J(x,R_n)  \ge g(x),\quad \forall n\in\N. 
\end{align*}
As $n\to\infty$, we deduce from \eqref{dominated} that
\begin{align*}
\alpha \inf_{\P\in\cP(x)}\E^{\P}[e^{-r\rho^*}g(B_{\rho^*})] + (1-\alpha) \liminf_{n\to\infty}\sup_{\P\in\cP(x)} \E^{\P}[e^{-r\rho^n}g(B_{\rho^n})] \ge g(x). 
\end{align*}
Hence, to prove \eqref{to show}, it remains to show that 
\begin{equation}\label{to show'}
\sup_{\P\in\cP(x)} \E^{\P}[e^{-r\rho^*}g(B_{\rho^*})] \ge \liminf_{n\to\infty}\sup_{\P\in\cP(x)} \E^{\P}[e^{-r\rho^n}g(B_{\rho^n})].  
\end{equation}

Thanks to \eqref{condition}, for any $n\ge n^*$, 
\begin{align*}
\Big|e^{-r\rho^n}&g(B_{\rho^n})- e^{-r\rho^*}g(B_{\rho^*})\Big|= \left|e^{-r\rho^n}g(B_{\rho^n})- e^{-r\rho^*}g(B_{\rho^*})\right| 1_{\{\rho^*<\infty\}}\\ 
&\le \left(e^{-r\rho^n}\left|g(B_{\rho^n})-g(B_{\rho^*})\right|1_{\{\rho^n<\infty\}}+|g(B_{\rho^*})||e^{-r\rho^*}-e^{-r\rho^n}|\right) 1_{\{\rho^*<\infty\}}\\
&\le \kappa\left(|B_{\rho^n}-B_{\rho^*}|1_{\{\rho^n<\infty\}}\right) + C (\rho^n-\rho^*), %1_{\{\rho^*<\infty\}}, 
\end{align*}
where $\kappa:\R_+\to\R_+$ is a modulus of continuity of $g$ on the domain $([p-\eta,p]\cup[q,q+\eta])\cap I$, and $C>0$ is a constant independent of $n$, thanks to the boundedness of $g(B_{\rho^*})$ and the Lipschitz continuity of $x\mapsto e^{-r x}$ on $[0,\infty)$. Fix $\eps>0$. Take $\delta>0$ such that $\kappa(z)< \eps/2$ for $z<\delta$. Then,
\begin{align*}
&\P\left(\Big|e^{-r\rho^n}g(B_{\rho^n})- e^{-r\rho^*}g(B_{\rho^*})\Big|\ge \eps\right)\\
&\le \P\left(  \kappa\left(|B_{\rho^n}-B_{\rho^*}|1_{\{\rho^n<\infty\}}\right) + C (\rho^n-\rho^*)\ge \eps  \right)\\
&\le \P\left( |B_{\rho^n}-B_{\rho^*}|1_{\{\rho^n<\infty\}} \ge \delta\right)  + \P\left( \rho^n-\rho^*\ge \frac{\eps}{2C}  \right),\quad \forall \P\in\cP(x).
\end{align*}
By \eqref{capty1} and \eqref{capty2}, this implies
\begin{equation}\label{in ca}
\lim_{n\to\infty}\sup_{\P\in\cP(x)}\P\left(\Big|e^{-r\rho^n}g(B_{\rho^n})- e^{-r\rho^*}g(B_{\rho^*})\Big|\ge \eps\right) = 0. 
\end{equation}
That is, $e^{-r\rho^n}g(B_{\rho^n})$ converges to $e^{-r\rho^*}g(B_{\rho^*})$ in capacity, in the sense of Definition 3.4 of \cite{CoJiPeng}. Now, by Theorem 3.2 in \cite{CoJiPeng}, \eqref{in ca} and \eqref{bdd} together imply
\[
 \lim_{n\to\infty}\sup_{\P\in\cP(x)} \E^{\P}[e^{-r\rho^n}g(B_{\rho^n})]=\sup_{\P\in\cP(x)} \E^{\P}[e^{-r\rho^*}g(B_{\rho^*})].  
\]
Then, \eqref{to show} is verified, which completes the proof.
\end{proof}

%%%%%%%%%%%%%%%%%%%%%%%%%%%%%%%%%%%%%%
%%%%%%%%%%%%%%%%%%%%%%%%%%%%%%%%%%%%%%

\section{Application to Real Options Valuation}\label{sec:example}
Coined by \cite{Myers} and popularized by \cite{MS86}, real options valuation refers to applying financial option pricing techniques to corporate investment decision making. The essence is to evaluate the right, but not the obligation, to undertake certain business plan, such as initiating, abandoning, expanding, or contracting a capital investment project. An optimal stopping problem can be accordingly formulated, and its solution dictates optimal timing or scheduling of investment outlays.

According to the seminal monograph \cite{Dixit94}, real options valuation fall into two categories: 
(i) dynamic programming under the physical measure, and (ii) contingent claim analysis under the risk-neutral measure. These two methods, as explained in depth by \cite{Dixit94}, have their respective advantages and limitations, and are theoretically justified under different market conditions. The large literature on real options follow these two methods closely, including (but not limited to) \cite{MS86},  \cite{Trigeorgis91} and \cite{Brandao05} under method (i), and \cite{Smith95}, \cite{Hugonnier07}, and \cite{Schwartz13} under method (ii). 

%This section is devoted to enhancing method (ii) by incorporating model ambiguity and ambiguity attitude. 
By nature, real options valuation may suffer model ambiguity {\it more} severely than pricing a typical financial option: as the underlying asset of a real option may not be tradable or fully observable, determining its dynamics
%a risk-neutral measure 
relies largely on an agent's estimate and belief. This often leads to %a collection of plausible risk-neutral measures and a corresponding 
an interval of plausible values of a real option. 
How to deal with these multiple values is unclear in the literature. Standard investment models assume that agents are completely ambiguity-averse, considering solely the worst case, i.e., the least value of the real option; see e.g., \cite{NO07}, \cite{TroKort}, and \cite{MiaoWang}. On the other hand, many empirical studies, including \cite{HeathTversky} and \cite{Bhide-book-99}, suggest heterogeneous ambiguity attitude, towards the same investment opportunities, among investors---some can be quite ambiguity-loving. 

In this section, we incorporate the $\alpha$-maxmin preference into real options valuation. This yields an immediate benefit: $\alpha\in [0,1]$, which measures an agent's ambiguity aversion, turns the multiple values of a real option into one, i.e., the convex combination of the least and the best values, weighted by $\alpha$ and $1-\alpha$, respectively. There is, however, a downside of it: the decision making problem now becomes time-inconsistent. Note that a related stopping problem under the $\alpha$-maxmin preference was introduced, but {\it not} solved,  in \cite{Schroder}, precisely because of the time inconsistency involved. By contrast, we will resolve a practical real options valuation problem under the $\alpha$-maxmin preference, on strength of the developments in Section~\ref{sec:model} and \ref{sec:converg}: {\it all} equilibria, as well as the {\it best} one among them, will be fully characterized under appropriate conditions.
%As a major contribution of this paper, we will provide complete characterizations of not only 

Specifically, we take the underlying asset $X$ to be a geometric Brownian motion, i.e., 
\begin{align}\label{GBM}
 X^{x,b,\sigma}_t = x + \int_0^t b X^{x,b, \sigma}_s ds +  {\int_{0}^{t}} \sigma X^{x,b,\sigma}_s dB_s,\quad \forall t\ge 0,\qquad \P_0\hbox{-a.s.},
\end{align}
for some $b\in\R$ and $\sigma>0$, yet an investor is uncertain about the true values of $b$ and $\sigma$. Following the uncertain volatility model in \cite{ALP95} and \cite{Lyons95}, we assume $\underline{\sigma}\leq \sigma\leq \overline{\sigma}$ for some known constants $0<\underline{\sigma}<\overline{\sigma}$. We also allow for uncertain drift, assuming $\underline{b}\leq b\leq \overline{b}$ for some known constants $\underline{b}<\overline{b}$.
This gives rise to a collection of plausible probability measures: similarly to \eqref{P^x}, each $\sigma\in[\underline{\sigma},\overline{\sigma}]$ and $b\in[\underline{b},\overline{b}]$ correspond to $\P^x_{\sigma,b} := \P_0\circ (X^{x,b,\sigma})^{-1}$. 

%We also assume that there is a known riskfree rate $r> 0$.  

In this section, we focus on the payoff function $g(x):=(K-x)^+$ of the real option, for some given $K>0$. In view of the setup in Sections~\ref{sec:model} and \ref{sec:converg}, we have $I=(0,\infty)$ and the expected payoff \eqref{J} now takes the form
\begin{align}\label{J'}
J(x,R)&=\alpha \inf_{\sigma\in[\underline{\sigma},\overline{\sigma}], b\in[\underline{b},\overline{b}]} \E^{\P_0}\left[e^{-rT_R}(K-X^{x,b,\sigma}_{T_R})^+\right]\nonumber\\ 
&\hspace{0.2in}+ (1-\alpha) \sup_{\sigma\in[\underline{\sigma},\overline{\sigma}], b\in[\underline{b},\overline{b}]} \E^{\P_0}\left[e^{-rT_R}(K-X^{x,b,\sigma}_{T_R})^+\right],
\end{align}
where $T_R$ is defined similarly to \eqref{hitX} as $T_R := \inf\{t>0:X^{x,b,\sigma}_t \in R\}$. Our goal is to characterize all (closed) equilibria $R$, and find the best one $\hat R$ among them; recall from Remark~\ref{rem:closed E} that we can focus on closed equilibria in the current setting. To this end, we need to first introduce an optimality criterion for an equilibrium. For any $R\in \cE$, we define % the value function 
\[
V(x,R) := g(x)\vee J(x,R),\quad\forall x\in I.  
\]

\begin{definition}\label{def:optimal E}
$\hat{R}\in\mathcal{E}$ is called an optimal equilibrium, if for any $R\in\cE$, we have
\begin{align*}
V(x,\hat{R})\geq V(x,R),\ \ \forall x\in {I}.
\end{align*}
\end{definition}
This criterion, introduced in \cite{HZ17-discrete}, is rather strong: it requires a subgame perfect Nash equilibrium to dominate any other equilibrium on the entire state space. For stopping problems under non-exponential discounting, \cite{HZ17-discrete, HZ17-continuous} establish the general existence of an optimal equilibrium, when the discount function induces decreasing impatience. In an example of optimal stopping under probability distortion, \cite{HNHZ18} derive an optimal equilibrium; see Section 4.3 therein. For the current real options valuation problem under model ambiguity, we will show that an optimal equilibrium also exists under appropriate conditions.

\begin{remark}
For time-inconsistent stopping problems, an equilibrium can be defined as in the present paper (i.e., Definition~\ref{def:E}, based on the fixed-point approach in \cite{HN18}), as in \cite{CL18} (based on the standard definition of an equilibrium for control problems in \cite{EL06}), or as in \cite{BZZ} (based on ``strong equilibria'' for control problems in \cite{HZ18-MOR}). As argued in \cite{BZZ} and \cite{HZ18-MOR}, the third kind of definition captures the idea of subgame perfect Nash equilibrium most accurately: it prevents deviation from the present strategy in a however small time interval starting from today---an ideal property that may not be achieved by an equilibrium of the first or the second kind. %under time inconsistency. 
 In a continuous-time Markov chain model, \cite{BZZ} analyze these three types of equilibria in detail, showing that an optimal equilibrium of the first kind (i.e., defined as in Definition~\ref{def:optimal E}) is in fact an equilibrium of the third kind. If such a result could be generalized to a diffusion model (which remains an open problem), an optimal equilibrium in Definition~\ref{def:optimal E} would automatically possess the ideal property mentioned above. %prevent deviation in a however small time interval starting from today. 
 \end{remark}

Let us start with characterizing closed equilibria that are contained in $(0,K]$. It will be shown in the end that this focus on $(0,K]$ is not restrictive at all. 

\begin{lemma}\label{chareq-1}
Suppose $\underline b \ge 0$. For any $R\in\mathcal{E}$ that is closed and contained in $(0,K]$, $R=(0,a]$ for some $a\in(0,K]$.
\end{lemma}

\begin{proof}
Define $a:=\sup\{x:x\in R\}\leq K$. By contradiction, suppose that there exists $x\in(0,a)$ such that $x\notin R$. Consider
\begin{align}\label{p,q}
p:=\sup\{y\in R:y<x\}\ \ \text{and}\ \ \ q:=\inf\{y\in R: y>x\},
\end{align}
where the supremum is taken to be $0$ if there exists no $y\in R$ such that $y<x$. By the closedness of $R$, we have $p<x<q$ and hence $\rho_R>0$ $\mathbb{P}_0$-a.s. In view of \eqref{J'}, this implies
\begin{align}
J(x,R)
&<\alpha  \inf_{\sigma\in[\underline{\sigma},\overline{\sigma}], b\in[\underline{b},\overline{b}]} \E^{\P_0}[K-X^{x,b,\sigma}_{T_R}]+ (1-\alpha)  \sup_{\sigma\in[\underline{\sigma},\overline{\sigma}], b\in[\underline{b},\overline{b}]}\E^{\P_0}[K-X^{x,b,\sigma}_{T_R}]\notag\\
&=\alpha  \inf_{\sigma\in[\underline{\sigma},\overline{\sigma}], b\in[\underline{b},\overline{b}]} \left(K-\E^{\P_0}[X^{x,b,\sigma}_{T_R}]\right)+ (1-\alpha)  \sup_{\sigma\in[\underline{\sigma},\overline{\sigma}], b\in[\underline{b},\overline{b}]}\left(K-\E^{\P_0}[X^{x,b,\sigma}_{T_R}]\right)\notag\\
&\le \alpha(K-x)+(1-\alpha)(K-x)=K-x=g(x),\label{sub}
\end{align}
where the last inequality follows from $X^{x,b,\sigma}$ being a $\P_0$-submartingale for all $\sigma\in[\underline{\sigma},\overline{\sigma}]$ and $b\in[\underline{b},\overline{b}]$, thanks to $\underline b\ge 0$. It follows that $x\in S_R$, a contradiction to $R$ being an equilibrium.
\end{proof}

%Note that for a closed equilibrium $R\in \cE$, there are three possibilities: (i) $R$ is one-sided, i.e. an element of
%\begin{align*}
%\cE_1:= \{R\in\cE: R= (0,a]\ \hbox{for some}\ a>0\}.
%\end{align*}
%Note that $[a,\infty)$ cannot be an equilibrium for any $a>0$, as $J(x,[a,\infty))< (K-a)^+< (K-x)^+=g(x)$ for all $0<x<a$. 

To obtain the converse of Lemma~\ref{chareq-1}---for which $a > 0$ the set $R = (0, a]$ is an equilibrium---
requires a detailed analysis on the map $x\to J(x,(0, a])$. For each $a\in (0,K)$, we define
\begin{align*}
&\Lambda(x,a):= J(x,(0,a])\\
&= (K-a) \left(\alpha\inf_{\sigma\in[\underline{\sigma},\overline{\sigma}],b\in[\underline{b},\overline{b}]} \E^{\P_0}\left[e^{-rT_a^{x,b,\sigma}}  \right] +(1-\alpha) \sup_{\sigma\in[\underline{\sigma},\overline{\sigma}],b\in[\underline{b},\overline{b}]}\E^{\P_0}\left[e^{-rT_a^{x,b,\sigma}}  \right]\right),\ \ \text{for}\ a\leq x<\infty,
\end{align*}
where $T_a^{x,b,\sigma}$ is defined as in \eqref{hitX}. Thanks to the formula on p.628 of \cite{BS-book-02},
\begin{equation}\label{BS}
\E^{\P_0}\left[e^{-rT_a^{x,b,\sigma}}  \right] = \left(\frac{a}{x}\right)^{\sqrt{\left(\frac{b}{\sigma^2}-\frac{1}{2}\right)^2+\frac{2r}{\sigma^2}}+\frac{b}{\sigma^2}-\frac{1}{2}}. % = \left(\frac{a}{x}\right)^{\frac{2r}{\sigma^2}}.
\end{equation}
It can be checked by direct calculation that the map 
\[
(b,\sigma)\mapsto \sqrt{\left(\frac{b}{\sigma^2}-\frac{1}{2}\right)^2+\frac{2r}{\sigma^2}}+\frac{b}{\sigma^2}-\frac{1}{2}
\]
is strictly increasing in $b$, and strictly decreasing in $\sigma$. It follows that
\begin{align}\label{LAMD}
\Lambda(x,a)
&=(K-a) \left(\alpha \left(\frac{a}{x}\right)^{m_1} +(1-\alpha) \left(\frac{a}{x}\right)^{m_2}\right),\quad \text{for}\ a\leq x<\infty,
\end{align}
where
\begin{align}\label{m_1,2}
m_1:=\sqrt{\left(\frac{\overline{b}}{\underline{\sigma}^2}-\frac{1}{2}\right)^2+\frac{2r}{\underline{\sigma}^2}}+\frac{\overline{b}}{\underline\sigma^2}-\frac{1}{2}>0\quad  \text{and}\quad m_2:=\sqrt{\left(\frac{\underline{b}}{\overline{\sigma}^2}-\frac{1}{2}\right)^2+\frac{2r}{\overline{\sigma}^2}}+\frac{\underline{b}}{\overline\sigma^2}-\frac{1}{2}>0.
\end{align}
Let us also introduce
\begin{equation}\label{a^*}
a^*:=\frac{m_1\alpha+m_2(1-\alpha)}{1+m_1\alpha+m_2(1-\alpha)} K\in (0,K).
\end{equation}
The next result collects useful properties of $\Lambda(x,a)$.

\begin{lemma}\label{char-1}
The function $\Lambda:\{(x,a)\in (0,\infty)\times(0,K]:x>a\}\to \R$ in \eqref{LAMD} satisfies the following properties. First, for any $a\in (0,K)$, %For $a\leq x<\infty$, we have
\begin{itemize}
\item[(i)] $x\mapsto\Lambda(x,a)$ is strictly decreasing and strictly convex on $(a,\infty)$, with $\Lambda(a,a)=K-a$ and $\lim_{x\rightarrow \infty}\Lambda(x,a)=0$;
\item[(ii)] if $a<a^*$, the two functions $x\mapsto\Lambda(x,a)$ and $x\mapsto (K-x)^+$ intersect exactly once at some $x^*\in (a,K)$, with $\Lambda(x,a)<(K-x)^+$ on $(a,x^*)$ and $\Lambda(x,a)>(K-x)^+$ on $(x^*,\infty)$;
\item[(iii)] if $a\geq a^*$, then $\Lambda(x,a)>(K-x)^+$ on $(a,\infty)$.
\end{itemize}
Moreover, for any $x\ge a^*$, 
\begin{itemize}
\item [(iv)] $a\mapsto\Lambda(x,a)$ is strictly decreasing on $(a^*,x\wedge K)$.
\end{itemize}
\end{lemma}

\begin{proof}
It can be checked directly from \eqref{LAMD} that (i) holds. For (ii) and (iii), it suffices to check the slope of $\Lambda(x,a)$ at $x=a$. Because
\begin{align*}
\lim_{x\downarrow a}\Lambda_x(x,a)&=\lim_{x\downarrow a} -\frac{K-a}{x}\left(m_1\alpha\left(\frac{a}{x}\right)^{m_1}+m_2(1-\alpha)\left(\frac{a}{x}\right)^{m_2}\right)\\
&=-\frac{K-a}{a} \left(m_1\alpha+m_2(1-\alpha)\right),
\end{align*}
we have $\lim_{x\downarrow a}\Lambda_x(x,a)<-1$ if and only if $a<a^*$. Now, with the properties in (i),  if $a<a^*$, $\lim_{x\downarrow a}\Lambda_x(x,a)<-1$ implies that $\Lambda(x,a)$ intersects $(K-x)^+$ exactly once at some $x^*\in (a,K)$; if $a\geq a^*$, $\lim_{x\downarrow a}\Lambda_x(x,a)\ge -1$ implies that $\Lambda(x,a)$ is always above $(K-x)^+$ on $(a,\infty)$.

To prove (iv), fix $x\ge a^*$. In view of \eqref{LAMD}, for any $a\in (a^*,x\wedge K)$,
\begin{equation}\label{LAMD' in a}
\Lambda_a(x,a) = -\frac1a \left[(1-\alpha) \big(a-m_2(K-a)\big) \left(\frac{a}{x}\right)^{m_2} + \alpha \big(a-m_1(K-a)\big) \left(\frac{a}{x}\right)^{m_1} \right].
\end{equation}
As $a> a^*$,
\begin{align}
a-m_2(K-a) &> \frac{(m_1\alpha+m_2(1-\alpha))K - m_2 (K-a) (1+m_1\alpha+m_2(1-\alpha))}{1+m_1\alpha+m_2(1-\alpha)}\notag\\
& > \frac{(m_1\alpha+m_2(1-\alpha))K - m_2 K}{1+m_1\alpha+m_2(1-\alpha)} = \frac{\alpha (m_1-m_2) K}{1+m_1\alpha+m_2(1-\alpha)},\label{for m2}
\end{align}
where the second line follows from $(m_1\alpha+m_2(1-\alpha))K< (1+m_1\alpha+m_2(1-\alpha)) a$, equivalent to $a> a^*$. 
A similar calculation yields
\begin{equation}\label{for m1}
a-m_1(K-a) >  \frac{(m_1\alpha+m_2(1-\alpha))K - m_1 K}{1+m_1\alpha+m_2(1-\alpha)}= \frac{-(1-\alpha)(m_1-m_2) K}{1+m_1\alpha+m_2(1-\alpha)}.
\end{equation}
By \eqref{for m2} and \eqref{for m1}, \eqref{LAMD' in a} leads to
\[
\Lambda_a(x,a) < -\frac{\alpha(1-\alpha)(m_1-m_2)K}{a (1+m_1\alpha+m_2(1-\alpha))} \left[\left(\frac{a}{x}\right)^{m_2} - \left(\frac{a}{x}\right)^{m_1} \right].
\]
As $m_1, m_2>0$ and $\frac{a}{x}<1$ for $a\in (a^*,x\wedge K)$, the above implies $\Lambda_a(x,a)<0$, as desired.
\end{proof}

A complete characterization of closed equilibria contained in $(0, K]$ can now be established. 

\begin{proposition}\label{characterization}
Suppose $\underline b \ge 0$. Then, 
$\cE_{(0,K]} := \{(0,a]:a^*\le a\le K\}$ is the collection of all closed equilibria contained in $(0,K]$. Moreover, %$\hat R:= (0,a^*]$ is optimal in $\cE_{(0,K]}$.
for any $a^*<a\le K$, 
\[
J(x,(0,a^*]) > J(x,(0,a])\quad \hbox{for all}\ x>a^*. 
\]
%Recall $a^*\in (0,K)$ defined in \eqref{a^*}. We have
%\begin{itemize}
%\item [(i)] $\cE_1 = \{(0,a]:a\ge a^*\}$.
%\item [(ii)] $\hat R:= (0,a^*]$ is optimal in $\cE_1$. Specifically, for any $a> a^*$, 
%\[
%J(x,(0,a^*]) > J(x,(0,a])\quad \hbox{for all}\ x>a^*. 
%\]
%\end{itemize}
\end{proposition}

\begin{proof}
In view of Lemma~\ref{chareq-1}, to prove the first assertion, it suffices to show that $(0,a]\in\mathcal{E}$ if and only if $a\geq a^*$. Observe that $(0,a]\in\mathcal{E}$ if and only if $J(x,(0,a])\ge g(x) = (K-a)^+$ for all $x>a$. As $J(x,(0,a])=\Lambda(x,a)$, Lemma \ref{char-1} asserts that this holds if and only if $a\ge a^*$. 

Set $\hat R:= (0,a^*]$ and take an arbitrary $R=(0,a]$ with $a^*<a\le K$. %By definition, $\hat R:= (0,a^*]$ satisfies $J(x,\hat R)=K-x=J(x,R)$ for all $x\in(0,a^*]$. 
For any $a^*<x\le a$, Lemma~\ref{char-1} (iii) implies $J(x,\hat R)=\Lambda(x,a^*) >K-x=J(x,R)$. For any $x>a$, Lemma~\ref{char-1} (iv) implies $J(x,\hat R)=\Lambda(x,a^*)>\Lambda(x,a)=J(x,R)$. Hence, we conclude that $J(x,\hat R)> J(x,R)$ for all $x>a^*$. %, and thus $V(x,\hat R)\ge V(x,R)$ for all $x\in I$.  
\end{proof}

Now, we show that focusing on equilibria contained in $(0,K]$ is by no means restrictive.

\begin{lemma}\label{lefteq}
Suppose $\underline b \ge 0$. For any $R\in\mathcal{E}$ that is closed, set $\bar a:=\sup\{x\in R:x\le K\}$. Then, we have $R\cap (0,K] = (0,\bar a]\in\cE$ and $J(x,(0,\bar a])\geq J(x,R)$ for all $x\in I$.
\end{lemma}

\begin{proof}
Note that $R\cap (0,K]\neq\emptyset$ must hold. If not, we would have $J(x,R)=0 < K-x = g(x)$ for all $0<x<K$, a contradiction to $R\in\mathcal{E}$. Hence, $\bar a$ is well-defined with $0<\bar a\le K$. %We then conclude $(0,\bar a]\in\cE$ from Proposition~\ref{characterization}.

To show $R\cap (0,K]=(0,\bar a]$, assume to the contrary that there exists $x\in (0,\bar{a})$ such that $x\notin R$. Similarly to the proof of Lemma \ref{chareq-1}, by considering $p$ and $q$ as in \eqref{p,q} and carrying out the calculation as in \eqref{sub}, we get $J(x,R)<K-x=g(x)$, a contradiction to $R\in\mathcal{E}$. 

To show $(0,\bar a]\in\cE$, it suffices to prove $\bar a\ge a^*$, thanks to Proposition~\ref{characterization}. Assume to the contrary that $\bar a< a^*$. Consider $\bar q:=\inf\{x\in R: x> K\}\ge K$. For any $x\in (\bar a,\bar q)$, note that 
\begin{align}
\E^{\P_0}\left[e^{-rT_{R}}g(X^{x,b,\sigma}_{T_{R}})\right] &=  \E^{\P_0}\left[e^{-r T^{x,b,\sigma}_{\bar a}}g(\bar a)1_{\{T^{x,b,\sigma}_{\bar a}< T^{x,b,\sigma}_{\bar q}\}}\right] \notag\\
&\le \E^{\P_0}\left[e^{-r T^{x,b,\sigma}_{\bar a}}g(\bar a)\right] = \E^{\P_0}\left[e^{-rT_{(0,\bar a]}}g(X^{x,b,\sigma}_{T_{(0,\bar a]}})\right],\quad \forall \sigma>0,\  b\in\R.\label{all sigma}
\end{align}
Hence, $J(x,R)\le J(x,(0,\bar a])$ for all $x\in (\bar a,\bar q)$. By Lemma~\ref{char-1} (ii), $\bar a < a^*$ implies that there exists $\delta>0$ small enough such that $J(x,(0,\bar a])=\Lambda(x,\bar a)< (K-x)^+ = g(x)$ for $x\in (\bar a,\bar a+\delta)$. Thus, we have $J(x,R)\le J(x,(0,\bar a])< g(x)$ for $x\in (\bar a,\bar a+\delta)$, a contradiction to $R\in\mathcal E$. 

To show the last assertion, note that if $\bar a=K$, it holds trivially that $J(x,R)=J(x,(0,\bar a])$ for all $x\in I$. Now, assume $\bar a<K$, and consider $\bar q$ as above. Clearly, $J(x,R)=J(x,(0,\bar a])$ for all $x\in I\setminus (\bar a,\bar q)$. For any $x\in (\bar a,\bar q)$, by the same calculation as in \eqref{all sigma}, we get 
%where $T^{x,r,\sigma}_{\bar a} :=\inf\{t>0:X^{x,r,\sigma}=\bar a\}$ and $T^{x,r,\sigma}_{\bar b} :=\inf\{t>0:X^{x,r,\sigma}=\bar b\}$. 
 $J(x,R)\le J(x,(0,\bar a])$. We then conclude that $J(x,R)\le J(x,(0,\bar a])$ for all $x\in I$. 
\end{proof}

Lemma \ref{lefteq} indicates that every closed equilibrium is dominated by another one contained in $(0,K]$. Consequently, in terms of finding an optimal equilibrium, it is enough to focus on $(0,K]$. This, together with Lemma~\ref{chareq-1} and Proposition~\ref{characterization}, immediately yields the following.

\begin{theorem}\label{thm:optimal E}
Suppose $\underline b \ge 0$. Then, $\hat{R}:=(0,a^*]$, with $a^*$ as in \eqref{a^*}, is an optimal equilibrium.
\end{theorem}

\begin{remark}
If one views $a^*$ in \eqref{a^*} as a function in $\alpha\in [0,1]$, it can be easily checked that $a^*$ is strictly increasing. That is, the larger $\alpha$ (i.e., the more ambiguity-averse), the larger the optimal equilibrium $(0,a^*]$. Intuitively speaking, if an agent is rather ambiguity-averse (i.e., with a large $\alpha$), he has strong intention to withdraw from the ambiguous environment---by stopping, in our current context. Hence, he prefers a large stopping threshold $a^*$, so that he can stop quickly once $X$ drifts only slightly below $K >0$, which yields a positive (yet small) payoff $K- a^*$. On the other hand, if an agent is rather ambiguity-loving (i.e., with a small $\alpha$), he has strong intention to stay in the ambiguous environment, to fully exploit the downward potential of $X$. Hence, he delays stopping by choosing a small stopping threshold $a^*$. 
\end{remark}

\begin{remark}
Recall that ``contingent claim analysis under the risk-neutral measure'' is one of the two major frameworks for real options valuation; see the second paragraph of this section. This framework stipulates that the discount rate $r>0$ should be the riskfree rate and the drift of $X$ in \eqref{GBM} should be $b=r$, in line with the usual risk-neutral pricing procedure. Hence, drift uncertainty does not play a role here with $\underline b =\overline b =r$. This largely simplifies $m_1, m_2$ in \eqref{m_1,2} to 
\begin{equation*}%\label{m_1,2'}
m_1=\frac{2r}{\underline\sigma^2}\quad \hbox{and}\quad m_2=\frac{2r}{\overline\sigma^2}.
\end{equation*} 
In particular, this shows that Theorem~\ref{thm:optimal E} is consistent with the standard risk-neutral pricing result without ambiguity. Indeed, if we additionally have $\underline\sigma=\overline\sigma=\sigma>0$ (i.e., no volatility uncertainty either), then $m_1=m_2 = \frac{2r}{\sigma^2}$. It follows that $a^*$ in \eqref{a^*} reduces to 
\[
a^* = \frac{{2r}/{\sigma^2}}{1+{2r}/{\sigma^2}} K.
\]
This is exactly the optimal stopping threshold for the classical pricing problem of a perpetual American put, i.e., 
\[
\sup_{\tau\in\T}\E^{\P_0}[e^{-r\tau} (K-X^{x,r,\sigma}_\tau)^+];
\] 
see e.g., Theorem 2.7.2 in \cite{KS-book-98}.
\end{remark}
 
%Estimating the riskfree rate $r>0$ is essential to real options valuation. 
Our analysis can easily accommodate additional uncertainty in the discount rate.

\begin{remark}\label{rem:uncertain r}
Suppose that $r>0$ is only known to lie in $[\underline r, \overline r]$, for some given $0< \underline r<\overline r<\infty$. %Recall that in the current risk-neutral pricing framework, the drift coefficient $b$ in \eqref{GBM} is taken to be $r$. Hence, the uncertainty in $r$ affects both discounting and state dynamics. 
By taking $r=\overline r$ in $m_1$ and $r=\underline r$ in $m_2$ in \eqref{m_1,2}, all subsequent analysis still holds. That is, Theorem~\ref{thm:optimal E} remains true: $(0,a^*]$ is an optimal equilibrium, where $a^*$ is defined as in \eqref{a^*} with the updated $m_1$ and $m_2$. 
\end{remark}

%to incorporate both drift $b\neq r$ and volatility uncertainty. Without loss of generality, let us fix the riskfree rate $r>0$.

%\begin{remark}%{Extension to Riskfree Rate Uncertainty}
%Suppose that $r>0$ is given, and the drift coefficient $b$ in \eqref{GBM} is only known to lie in $[\underline b, \overline b]$, for some given constants $0\le \underline b<\overline b<\infty$. Unlike the risk-neutral setup in Remark~\ref{rem:uncertain r}, the uncertainty here affects state dynamics, but not discounting. By taking 
%\begin{align*}
%m_1:=\sqrt{\frac{\overline{b}^2}{\underline{\sigma}^4}+\frac{(2r-\overline{b})}{\underline{\sigma}^2}+\frac{1}{4}}+\frac{\overline{b}}{\underline{\sigma}^2}-\frac{1}{2}\quad  \text{and}\quad \ \ m_2:=\sqrt{\frac{\underline{b}^2}{\overline{\sigma}^4}+\frac{(2r-\underline{b})}{\overline{\sigma}^2}+\frac{1}{4}}+\frac{\underline{b}}{\overline{\sigma}^2}-\frac{1}{2}
%\end{align*}
%in \eqref{a^*}, all subsequent analysis still holds, leading to a corresponding version of Theorem~\ref{thm:optimal E}: $(0,a^*]$ is an optimal equilibrium, where $a^*$ is defined as in \eqref{a^*} with the updated $m_1$ and $m_2$. 
%\end{remark}

\begin{remark}
It is of interest to investigate if Theorem~\ref{thm:optimal E} still holds when $\underline b<0$. Preliminary studies indicate that the same analysis in this section is inadequate. Indeed, with $\underline b<0$, $X$ in \eqref{GBM} is no longer guaranteed a submartingale. Consequently, the convenient characterization of equilibria as one-sided intervals (Lemma~\ref{chareq-1}) no longer holds.  
Specifically, through numerical experiments, we find that under $\underline b<0$, there may exist ``two-sided'' equilibria, i.e., $(0,p]\cup[q,\infty)\in\cE$ for some $0<p<q$. A few examples include
\begin{itemize}
\item By taking $r=0.2$, $\underline{b}=-8$, $\overline{b}=-2$, $\underline{\sigma}=0.2$, $\overline{\sigma}=0.8$, $K=10$, and $\alpha=0.9$, we get $m_{1}\approx0.0989$, $m_{2}\approx0.0240$, and $a^*\approx0.8375$. Moreover, numerical computation indicates
%\[
 %(0,2]\cup[5,\infty)\in\cE\quad \hbox{and}\quad 
$(0,p]\cup[q,\infty)=(0,0.93]\cup[5,\infty)\in\cE.$
%\]
\item By taking $r = 0.7$, $b = -10$, $b = -2.5$, $\underline\sigma = 1.2$, $\overline\sigma = 4$, $K = 0.8$, and $\alpha=0.9$, we get $m_1 \approx 0.2077$, $m_2 \approx 0.0382$, and $a^*\approx0.1282$. Moreover, numerical computation indicates $ (0,p]\cup[q,\infty)=(0,0.3]\cup[0.6,\infty)\in\cE.$
\end{itemize}
In both examples, we have $(0,p]\cup[q,\infty)\in\cE$ with $p>a^*$. Note that ``$(0,a]\in\cE$ for all $a\ge a^*$'' still holds under $\underline b<0$, since $J(x,(0,a])=\Lambda(x,a)$ and Lemma~\ref{char-1} does not depend on the sign of $\underline b$. Now,  take $a\in (p,q)$ in the above two examples. Then, $(0,p]\cup[q,\infty)$ and $(0,a]$ are two equilibria that neither one dominates the other: $J(x, (0,p]\cup[q,\infty))> K-x = J(x,(0,a])$ for $x\in(p,a)$, while $J(x,(0,a]) = \Lambda(x,a)> (K-x)^+=J(x, (0,p]\cup[q,\infty))$ for $x\in(q,\infty)$. 

That is, under $\underline b<0$, equilibria can be one-sided or two-sided, and there need not be a dominant one among any two equilibria. This is in contrast to the case $\underline b\ge 0$ studied in this section, where equilibria are all one-sided and there must be a dominant one among two equilibria. In view of this, many arguments in this section do not apply to the case $\underline b<0$.
%As all subsequent arguments rely on this characterization of equilibria, they do not apply to the case $\underline b<0$. 
We expect that very different techniques are needed for $\underline b<0$, and would leave this for future research.
\end{remark}

%%%%%%%%%%%%%%%%%%%%%%%%%%%%%

\subsection{An Example: Non-Existence of an Optimal Equilibrium}
In view of Theorem~\ref{thm:optimal E}, it is natural to ask whether the existence of an optimal equilibrium is a general fact under model ambiguity. We show that this is {\it not} the case, by providing a counterexample. 
%For stopping problems under non-exponential discounting (with no model ambiguity), \cite{HZ17-discrete, HZ17-continuous} establish the general existence of optimal equilibria, on strength of a partial order among equilibria: for any two equilibria with one  contained in the other, the smaller one dominates the larger one. This partial order, derived in Lemma 3.1 of \cite{HZ17-continuous}, hinges on linearity of expectation under a single probability. In our case with model ambiguity, no clear partial order exists among equilibria, under the nonlinear expectation $\alpha \inf_{\P\in\cP(x)}\E^\P[\cdot] +(1-\alpha) \sup_{\P\in\cP(x)}\E^\P[\cdot]$. General existence of optimal equilibria is then in question. 
Specifically, we consider the state process $X$ in \eqref{GBM} and the same model ambiguity as specified below \eqref{GBM}. The payoff function to focus on is $g(x):=x$. 

For each $\sigma>0$, $b\in\R$, and $a>0$, consider the function %$\kappa:\{(x,a)\in (0,\infty)\times(0,\infty):x\le a\}\to \R$ defined by
\[
\kappa^{\sigma,b}(x,a):= a\cdot \E^{\P_0}\left[e^{-rT_a^{x,b,\sigma}}  \right] = a \left(\frac{x}{a}\right)^{\sqrt{\left(\frac{b}{\sigma^2}-\frac{1}{2}\right)^2+\frac{2r}{\sigma^2}}-\frac{b}{\sigma^2}+\frac{1}{2}},\quad \hbox{for}\ 0<x\le a, 
\]
where the second equality follows from the formula on p.628 of \cite{BS-book-02}. By definition, 
\begin{equation}\label{kappa}
\lim_{x\downarrow 0} \kappa^{\sigma,b}(x,a)=0\quad \hbox{and}\quad \kappa^{\sigma,b}(a,a)=a. 
\end{equation}
Also, it can be checked by direct calculation that $\kappa^{\sigma,b}_{xx}(x,a)<0$ if and only if $\sqrt{\left(\frac{b}{\sigma^2}-\frac{1}{2}\right)^2+\frac{2r}{\sigma^2}}-\frac{b}{\sigma^2}-\frac{1}{2}<0$, which is equivalent to $b>r$. That is, 
\begin{equation}\label{b>r}
\kappa^{\sigma,b}_{xx}(x,a)<0 \iff b>r.
\end{equation}

\begin{lemma}\label{lem:R_a in E}
Suppose $\underline b>r$. Then, $[a,\infty)\in\cE$ for all $a>0$.
\end{lemma}

\begin{proof}
Fix $a>0$ and set $R_a:=[a,\infty)$. For each $\sigma>0$ and $b>r$, observe that 
\begin{equation}
\E^{\P_0}\left[e^{-r \rho_{R_a}}X^{x,\sigma,b}_{\rho_{R_a}}\right] = \kappa^{\sigma,b}(x,a)> x,\quad \forall 0<x<a,
\end{equation} 
where the inequality follows from \eqref{kappa} and \eqref{b>r}. Hence, with $\underline b>r$, 
\begin{align}\label{J>x}
&J(x,R_a)= \alpha\inf_{\sigma\in[\underline{\sigma},\overline{\sigma}],b\in[\underline{b},\overline{b}]} \kappa^{\sigma,b}(x,a) +(1-\alpha) \sup_{\sigma\in[\underline{\sigma},\overline{\sigma}],b\in[\underline{b},\overline{b}]}\kappa^{\sigma,b}(x,a)>x,\quad \forall 0 <x<a.
\end{align}
This shows that $R_a\in \cE$. 
\end{proof}

The next result shows that certain configuration of model ambiguity (i.e., condition of $\underline b$, $\overline b$, $\underline \sigma$, and $\overline \sigma$) does not allow an optimal equilibrium to exist. 

\begin{proposition}\label{prop:no optimal E}
Suppose $\underline b>r$ and $\overline b<r+\frac{\underline{\sigma}^2}{2}$. Then, there exists no optimal equilibrium. 
\end{proposition}

\begin{proof}
By contradiction, suppose that there exists an optimal equilibrium $\overline R$. If $\overline R\neq\emptyset$, then for any $x\in\overline R$, we have $J(x,\overline R) = x < J(x,[a,\infty))$ for all $a>x$, where the inequality follows from \eqref{J>x}. However, as $[a,\infty)\in\cE$ for all $a>0$ by Lemma~\ref{lem:R_a in E}, the above violates the fact that $\overline R$ is an optimal equilibrium. Hence, we must have $\overline R=\emptyset$. This readily yields a contradiction, as $\emptyset$ is not even an equilibrium. Indeed, under $\underline b>r$ and $\overline b<r+\frac{\underline{\sigma}^2}{2}$, every $b\in[\underline b, \overline b]$ satisfies $r<b<r+\sigma^2/2$ for all $\sigma\in[\underline\sigma, \overline\sigma]$. Hence, for any $\sigma\in[\underline{\sigma},\overline{\sigma}]$ and $b\in[\underline{b},\overline{b}]$,
\[
e^{-rt} X^{x,\sigma,b}_t = x e^{(b-\frac{\sigma^2}{2}-r)t+\sigma W_t} \to 0\quad \hbox{as}\ t\to\infty,\quad \forall x>0.
\]
This, together with \eqref{at infty}, implies that for any $\sigma\in[\underline{\sigma},\overline{\sigma}]$ and $b\in[\underline{b},\overline{b}]$,
%\begin{equation}
$\E^{\P_0}[e^{-r \rho_{\emptyset}}X^{x,\sigma,b}_{\rho_\emptyset}] =0$ for all $x>0$.
%\end{equation} 
It follows that $J(x,\emptyset) = 0<x$ for all $x>0$, which shows that $\emptyset\notin\cE$. 
\end{proof}

\begin{remark}
When there is no model ambiguity (i.e., $\underline{\sigma}=\overline{\sigma}=\sigma$ and $\underline{b}=\overline{b}=b$), as long as $r<b<r+{\sigma^2}/{2}$, the same arguments in the proofs of Lemma~\ref{lem:R_a in E} and Proposition~\ref{prop:no optimal E} readily show that there exists no optimal equilibrium. In other words, the condition ``$\underline b>r$ and $\overline b<r+{\underline{\sigma}^2}/{2}$'' in  Proposition~\ref{prop:no optimal E} is used to maintain this non-existence result derived in the ambiguity-free setting.  

Despite this observation, it is worth noting that without ambiguity, there is no time inconsistency and one does not need to discuss equilibria at all. That is, a (non-)existence result of an optimal equilibrium, such as Proposition~\ref{prop:no optimal E},  is meaningful only under model ambiguity. 
\end{remark}

%%%%%%%%%%%%%%%%%%%%%%%%%%%%%%%%%%%%%%
%%%%%%%%%%%%%%%%%%%%%%%%%%%%%%%%%%%%%

\section{A Generalized Measurable Projection Theorem}\label{sec:proj}
A measurable projection theorem typically involves the product of two measurable spaces, and studies whether the projection of a measurable set in the product space is still measurable. Classical results, see e.g. Theorem 2.12 of \cite{Crauel-book-02} or Theorem III.23 of \cite{CV-book-77}, all require one of the two spaces to be a Borel space endowed with the Borel $\sigma$-algebra. As pointed out in Remark \ref{clremk}, this does not serve our needs in the proof of Lemma \ref{measuretheta}, where Borel measurability is elusive. This section is devoted to establish a new, generalized measurable projection theorem that accommodates {\it any} two general measurable spaces; see Theorem~\ref{thm:projection} below, one of the major contributions of this paper.  

Let us start with the notion of {\it separated} measurable spaces. 
Given a set $M$, a collection $\C$ of subsets $M$ is said to {\it separate the points of $M$} if for any distinct $y_1,y_2\in M$, there exists $A\in \C$ that contains exactly one of $y_1$ and $y_2$. %Let us now introduce the notion of ``separated measurable spaces'', 
The next definition is taken from Section 8.6 of \cite{Cohn-book-80}.

\begin{definition}
A measurable space $(M,\A)$ is said to be separated if $\A$ separates the points of $M$, and countably generated if there exists $\{A_i\}_{i\in\N}$ in $\A$ such that $\A = \sigma(\{A_i\}_{i\in\N})$. %is countably separated if a countable sub-collection of $\A$ separates the points of $M$.
%there exist $\{A_i\}_{i\in\N}$ in $\A$ such that
%\begin{itemize}
%\item [(i)] $\A$ is generated by $\{A_i\}_{i\in\N}$;
%\item [(ii)] $\{A_i\}_{i\in\N}$ separate the points in $M$, i.e. for any distinct $y_1,y_2\in M$, there exists $i\in \N$ such that $A_i$ %contains exactly one of $y_1$ and $y_2$.
%\end{itemize}
\end{definition}

\begin{remark}
If a measurable space $(M,\A)$ is countably generated, it can be shown that $\A$ separates the points of $M$ if and only if $\{A_i\}_{i\in\N}$ separates the points of $M$; see e.g. Lemma III.24 of \cite{CV-book-77}.
Consequently, $(M,\A)$ being both separated and countably generated is the same as the notion ``separability'' defined in Definition III.24 of \cite{CV-book-77}.
\end{remark}

The benefit of $(M,\A)$ being separated and countably generated is that it can be analyzed much more easily---as if endowed with a Borel $\sigma$-algebra. This is stated precisely in the next result, taken from Proposition III.25 of \cite{CV-book-77} and Corollary 8.6.4 of \cite{Cohn-book-80}.

\begin{lemma}\label{lem:isomorphism}
Let $(M,\A)$ be a separated and countably generated measurable space. Then, there exists a subset $K$ of $\{0,1\}^\N$ such that $(M,\A)$ is isomorphic to $(K,\B(K))$.
\end{lemma}

On strength of Lemma~\ref{lem:isomorphism}, a generalized measurable projection theorem can be readily established, for the special case where the two measurable spaces are separated and countably generated. 

To state the result appropriately, let us introduce additional notation. Given a measurable space $(M,\A)$, we denote by $\A^\mu$ the augmentation of $\A$ by $\mu$-null sets, for any finite measure $\mu$ on $(M,\A)$. Let $\hat \A$ be the {\it universal completion} of $\A$, i.e.
\[
\hat \A := \bigcap\{\A^\mu : \mu\ \hbox{is a finite measure on $(M,\A)$}\}.
\]

\begin{lemma}\label{lem:with separated condition}
Let $(M_1,\A_1)$ and $(M_2,\A_2)$ be two measurable spaces that are separated and countably generated. For any $G\in \A_1\otimes\A_2$, its projection $\pr_{M_1}(G)$ belongs to $\hat \A_1$.
\end{lemma}

\begin{proof}
In view of Lemma~\ref{lem:isomorphism}, there exist isomorphisms $i_1:(M_1,\A_1)\to (K_1,\B(K_1))$ and $i_2:(M_2,\A_2)\to (K_2,\B(K_2))$, for some $K_1$, $K_2\subseteq \{0,1\}^N$. We then obtain a one-to-one correspondence, induced by the maps $i_1$ and $i_2$, between elements in $\A_1\otimes \A_2$ and those in $\B(K_1)\otimes\B(K_2)$. Moreover, by Lemma III.26 of \cite{CV-book-77}, $i_1$ is not only $(\A_1,\B(K_1))$-measurable, but $(\hat\A_1,\cG(K_1))$-measurable, where $\cG(K_1)$ denotes the $\sigma$-algebra generated by analytic subsets of $K_1$.  

Now, suppose that $G'\in \B(K_1)\otimes\B(K_2)$ corresponds to $G\in \A_1\otimes\A_2$. Then,
$\pr_{M_1}(G) = i_1^{-1}\left(\pr_{K_1}(G')\right)$. By Proposition 7.39 of \cite{BS-book-78}, $\pr_{K_1}(G')$ is an analytic subset of $K_1$. It follows that $i_1^{-1}\left(\pr_{K_1}(G')\right)\in \hat \A_1$.
\end{proof}

Extending Lemma~\ref{lem:with separated condition} to accommodate any two arbitrary measurable spaces requires the following technical result.

\begin{lemma}\label{lem:graph}
Let $(M_1,\A_1)$ and $(M_2,\A_2)$ be two measurable spaces. For any $G\in \A_1\otimes\A_2$, there exist $A\in \A_1$ and a set-valued function $\Phi:A\to \A_2$ %2^{M_2}$, taking values in $\A_2$, 
such that
\begin{itemize}
\item [(i)]  $G$ is the graph of $\Phi$;
\item [(ii)] for any $y,z\in A$ satisfying $1_C(y) = 1_C(z)$ for all $C\in\A_1$, we have $\Phi(y)=\Phi(z)$.
\end{itemize}
\end{lemma}

\begin{proof}
Consider the collection
\[
\Gamma :=\{G \in \A_1\otimes\A_2\ :\ \exists A\in\A_1\ \hbox{and $\Phi:A\to \A_2$ such that (i) and (ii) hold}\}.
\]
First, observe that $\Gamma$ includes all sets of the form $H= A\times B$, with $A\in \A_1$ and $B\in \A_2$. Indeed, the constant set-valued function $\Phi(y) := B$, for all $y\in A$, obviously has $H$ as its graph and satisfies (ii) in a trivial way. Now, we claim that $\Gamma$ is a $\sigma$-algebra. As argued above, $M_1\times M_2\in \Gamma$. Next, for any $G\in \Gamma$, take $A\in\A_1$ and $\Phi: A\to\A_2$ such that (i) and (ii) hold. Define the set-valued function $\Psi:M_1\to\A_2$ by
\[
\Psi(y) :=
\begin{cases}
(\Phi(y))^c,\quad &\hbox{if}\ y\in A;\quad\\
M_2,\quad &\hbox{if}\ y\in A^c.
\end{cases}
\]
As $\Phi$ satisfies (ii), so does $\Psi$ by definition. It can also be checked that the graph of $\Psi$ is $G^c$. This implies $G^c\in\Gamma$. Finally, for any $\{G_n\}_{n\in\N}$ in $\Gamma$, take $\{A_n\}_{n\in\N}$ in $\A_1$ and $\Phi_n: A_n\to\A_2$ such that $G_n$ is the graph of $\Phi_n$ and $\Phi_n$ satisfies (ii) for all $n\in\N$. Let $A:=\bigcup_{n\in\N}A_n\in\A_1$, and define the set-valued function $\tilde\Psi:A\to\A_2$ by
\[
\tilde\Psi(y) := \bigcup_{n\in\N,\  y\in A_n}\Phi_n(y),\quad y\in A.
\]
With $\Phi_n$ satisfying (ii) for all $n\in\N$, $\tilde\Psi$ by definition also satisfies (ii). It can also be checked that the graph of $\tilde\Psi$ is $\bigcup_{n\in\N}G_n$. This implies $\bigcup_{n\in\N}G_n\in\Gamma$. As $\Gamma$ is a $\sigma$-algebra containing $H=A\times B$ for all $A\in\A_1$ and $B\in\A_2$, we must have $\A_1\otimes\A_2 \subseteq \Gamma$, %$\Gamma=\A_1\otimes\A_2$, 
which yields the desired result.
\end{proof}

Now, we are ready to present the main result of this section. %generalized measurable projection theorem. 

\begin{theorem}\label{thm:projection}
Let $(M_1,\A_1)$ and $(M_2,\A_2)$ be two measurable spaces. For any $G\in \A_1\otimes\A_2$, its projection $\pr_{M_1}(G)$ belongs to $\hat \A_1$.
\end{theorem}

\begin{proof}
Fix $G\in \A_1\otimes\A_2$. Consider 
\[
\fC_i := \{\C_i\subseteq\A_i\ :\ \C_i\ \hbox{is a countably generated $\sigma$-algebra}\},\quad i=1,2.
\]
First, we claim that $G\in \C_1\otimes\C_2$ for some $\C_1\in\fC_1$ and $\C_2\in\fC_2$. 
Observe that
\begin{equation}\label{countably generated}
\A_1\otimes\A_2 = \bigcup\{\C_1\otimes \C_2\ :\ \C_1\in\fC_1,\ \C_2\in\fC_2\}.
\end{equation}
Indeed, as the right hand side of \eqref{countably generated} is a $\sigma$-algebra and it contains all sets of the form $H=A\times B$ with $A\in\A_1$ and $B\in\A_2$ (this is because $H \in \C_1\otimes\C_2$, for any $\C_1\in\fC_1$ that contains $A$ and any $\C_2\in\fC_2$ that contains $B$), we obtain the ``$\subseteq$'' relation in \eqref{countably generated}. Because the ``$\supseteq$'' relation is trivial, \eqref{countably generated} is established. Our claim is therefore proved.

%Take $\{E_i\}_{i\in\N}$ in $\C_1$ (resp. $\{F_i\}_{i\in\N}$ in $\C_2$) such that $\{E_i\}_{i\in\N}$ generate $\C_1$ (resp. $\{F_i\}_{i\in\N}$ generate $\C_2$).
Define an equivalence relation on $M_1$ as follows: for any $y,z\in M_1$,
\begin{equation}\label{equiv}
y\sim z\quad \hbox{if and only if}\quad 1_C(y) = 1_C(z)\ \hbox{for all $C\in\C_1$}.
\end{equation}
Set $M'_1 := M_1/\sim$, the quotient space induced by $\sim$, and define $\varphi_1:M_1\to M'_1$ by
\begin{equation}\label{varphi}
\varphi_1(y) = [y] := \{z\in M_1 : z\sim y\},\quad \forall y\in M_1.
\end{equation}
One can deduce from \eqref{equiv} and \eqref{varphi} that for any $C_1$, $C_2\in\C_1$,
\[
\varphi_1(C_1)\neq \varphi_1(C_2)\ \hbox{if}\ C_1\neq C_2\quad \hbox{and}\quad \varphi_1(C_1)\cap \varphi_1(C_2)=\emptyset\ \hbox{if}\ C_1\cap C_2=\emptyset.
\]
Let us check that $\C'_1 := \varphi_1(\C_1)$ is a $\sigma$-algebra on $M'_1$. First, $\emptyset=\varphi_1(\emptyset)\in\C'_1$. Also, for any $\{C'_i\}_{i\in\N}$ in $\C'_1$, there exist $\{C_i\}_{i\in\N}$ in $\C_1$ such that $C'_i =\varphi_1(C_i)$ for all $i\in\N$. Consequently, (i) $\bigcup_{i\in\N}C'_i = \bigcup_{i\in\N} \varphi_1(C_i) = \varphi_1(\bigcup_{i\in\N} C_i)\in \C'_1$, where the second equality follows from the definition of $\varphi_1$; %and the inclusion is due to $\bigcup_{i\in\N} C_i\in \C_1$; 
(ii) Because $\varphi_1(C_1)\cup\varphi_1(C_1^c) = M'_1$ and $\varphi_1(C_1)\cap\varphi_1(C_1^c) = \emptyset$, we have $(C'_1)^c = (\varphi_1(C_1))^c=\varphi_1(C_1^c)\in \C'_1$. Hence, we conclude that $\C'_1$ is a $\sigma$-algebra. 

Because $\varphi_1:M_1\to M'_1$ is a surjection and $\C_1$ is countably generated, $\C'_1=\varphi_1(\C_1)$ is again countably generated. Also, for any distinct $[y], [z]\in M'_1$, there exists $C\in \C_1$ such that $y\in C$ but $z\notin C$; that is, $\varphi_1(C)\in\C'_1$ contains $[y]$, but not $[z]$. This shows that $\C'_1$ separates the points of $M'_1$. Therefore, the measure space $(M'_1,\C'_1)$ is separated and countably generated. 

In a similar fashion, we can define an equivalence relation on $M_2$ as in \eqref{equiv}, with $\C_1$ replaced by $\C_2$. Then, $\varphi_2:M_2\to M'_2$ can be introduced as in \eqref{varphi}, with $M_1$ and $M'_1$ replaced by $M_2$ and $M'_2:= M_2/\sim$. The same argument above implies that $(M'_2,\C'_2)$, with $\C'_2 := \varphi_2(\C_2)$, is separated and countably generated.

Recall that $G\in\C_1\otimes\C_2$. %with $\C_1\in\fC_1$ and $\C_2\in\fC_2$. 
By Lemma~\ref{lem:graph}, there exist $C^*\in\C_1$ and a set-valued function $\Phi:C^*\to \C_2$, such that $G$ is the graph of $\Phi$ and $\Phi(y)=\Phi(z)$ whenever $y\sim z$. Note that $\Phi$ can be extended to the entire space $M_1$ by setting $\Phi(y)=\emptyset$ for $y\notin C^*$. Define $\psi_1: M'_1\to M_1$ as follows: for any $[y]\in M'_1$, let $\psi_1([y]):=z$ for some $z\in M_1$ with $z\sim y$. Then, we deduce from \eqref{equiv} that for any $C\in\C_1$, $\psi_1^{-1}(C) = \varphi_1(C)\in \C'_1$; that is, $\psi_1$ is $(\C'_1,\C_1)$-measurable. Define $\psi_2:M'_2\to M_2$ in the same manner: for any $[y]\in M'_2$, let $\psi_2([y]) : =z$ for some $z\in M_2$ with $z\sim y$. Similarly, $\psi_2$ is $(\C'_2,\C_2)$-measurable. Now, by considering $\varphi_2$ as a function from $\C_2$ to $\C'_2$, we introduce the set-valued function $\Psi$ from $M'_1$ to $\C'_2$:
\[
\Psi([y]) := \varphi_2\left(\Phi(\psi_1([y]))\right)\in \C'_2,\quad \forall [y]\in %\varphi_1(C^*)\subseteq 
M'_1.
\]
Let $H$ denote the graph of $\Psi$. Observe that
\begin{align*}
H &= \{([y],[z])\in M'_1\times M'_2\ :\ [z]\in\Psi([y])\}\\
 &= \{([y],[z])\in M'_1\times M'_2\ :\ \psi_2([z])\in\Phi(\psi_1([y]))\} \\
&= \{([y],[z])\in  M'_1\times M'_2\ :\ \left(\psi_1([y]),\psi_2([z])\right)\in G\}  \\
&= (\psi_1\times\psi_2)^{-1} (G) \in \C'_1\otimes \C'_2,
\end{align*}
where the second equality is deduced from \eqref{equiv}, with $\C_1$ replaced by $\C_2$. By Lemma~\ref{lem:with separated condition}, this implies $\pr_{M'_1}(H)\in \hat{\C'_1}$. Thanks to Lemma III.26 of \cite{CV-book-77}, $\varphi_1$ is not only $(\C_1,\C_1')$-measurable, but $(\hat\C_1,\hat\C_1')$-measurable. Hence, $\pr_{M_1}(G) = \varphi_1^{-1}\big(\pr_{M'_1}(H) \big)\in \hat \C_1\subseteq\hat \A_1$.
\end{proof}

%%%%%%%%%%%%%%%%%%%%%%%%%%%
%%%%%%%%%%%%%%%%%%%%%%%%%%%

\appendix

\section{Proof of Lemma~\ref{lem:in capacity}}\label{sec:proof of lem:in capacity}
Fix $x\in I$. Consider
\begin{align*}%\label{p,q}
p:=\sup\{y\in R_0:y<x\}\quad \text{and}\quad  q:=\inf\{y\in R_0: y>x\},
\end{align*}
where we take $p=\ell$ (resp. $q=r$) if there is no $y<x$ (resp. $y>x$) lying in $R_0$. Similarly, define
\begin{align*}%\label{p,q}
p_n:=\sup\{y\in R_n:y<x\}\quad \text{and}\quad  q_n:=\inf\{y\in R_n: y>x\},\quad \forall n\in\N.
\end{align*}
As $R_0= \bigcup_{n\in\N} R_n$ and $\{R_n\}_{n\in\N}$ is nondecreasing, we have $p_n\uparrow p$ and $q_n\downarrow q$. 

\begin{proof}[Proof of \eqref{rho converge}]
On the set $\{\rho^0=\infty\}$, $\rho^n=\rho^0=\infty$ for all $n\in\N$, and thus \eqref{rho converge} trivially holds. On the set $\{\rho^0<\infty\}$, $B_{\rho^0}=p$ or $q$. We assume $B_{\rho^0}=p$ without loss of generality. If $p\in R_0$, then $p\in R_n$ for all $n$ large enough. Consequently, $\rho^n=\rho^0$ for all $n$ large enough. If $p\notin R_0$, then $B$ has to enter the region $(\ell,p)$ immediately after $\rho^0$. As $p_n\uparrow p$, this implies $\rho^n\downarrow\rho^0$. Thus, \eqref{rho converge} is established.  
\end{proof}

\begin{proof}[Proof of \eqref{capty1}]
Fix $\eps>0$. First, note that if $p_n=p$ and $q_n=q$ for $n$ large enough, then $\rho^n=\rho^0$ on $\Omega^x$ for all $n$ large enough, whence \eqref{capty1} follows trivially. It remains to deal with the case where (i) $p_n$ is strictly increasing, or (ii) $q_n$ is strictly decreasing. 

For any $n\in\N$, define $A_n := \{\omega\in\Omega^x:|\rho^n-\rho^0|\ge \eps\}$. By \eqref{rho converge}, $(A_n)_{n\in\N}$ is nonincreasing and $\bigcap_{n\in\N}A_n=\emptyset$. For the case where (i) and (ii) both hold, observe that $\overline{A_n}\setminus A_n \subseteq F_n$, where
\begin{align*}
F_n := \{\omega\in\Omega^x:\rho^0<\infty,\ &B_t\in[p_n,q_n]\ \forall t\in(\rho^0,\rho^0+\eps),\\
 &\exists s\in(\rho^0,\rho^0+\eps)\ \hbox{s.t.}\ B_s= p_n\hbox{ or}\ q_n\},\quad \forall n\in\N. 
\end{align*}
By the definition of $F_n$, we have
\begin{equation}\label{F_n null}
\P^x_{b,\sigma}(F_n)= \P_0\left((X^{x,b,\sigma}_t)_{t\ge 0}\in F_n\right) = 0,  \quad \forall (b,\sigma)\in\Pi(x). 
\end{equation}
Indeed, as $\{(X^{x,b,\sigma}_t)_{t\ge 0}\in F_n\}$ consists of sample paths such that $T^{p^n}_{(\ell, p_n)}>0$ or $T^{q^n}_{(q_n,r)}>0$, it must be a $\P_0$-null set in view of Remark~\ref{lem:T^x_x=0}. Moreover, $F_n\cap F_m=\emptyset$ for all $n<m$, as $p_n$ is strictly increasing and $q_n$ is strictly decreasing. It follows that 
\begin{equation}\label{also empty}
\bigcap_{n\in\N} \overline{A_n} \subseteq \bigcap_{n\in\N} (A_n\cup F_n) = \bigcap_{n\in\N} {A_n}=\emptyset. 
\end{equation}
By Lemma 7 in \cite{DHP11}, as $\cP(x)$ is relatively compact (Lemma~\ref{lem:rc}), for every sequence of closed set $C_n\downarrow\emptyset$ we have $\sup_{\P\in\cP(x)}\P(C_n)\downarrow 0$. Hence, by \eqref{F_n null} and \eqref{also empty}, 
\begin{equation}\label{to 0 case 1}
\sup_{\P\in\cP(x)}\P({A_n}) = \sup_{\P\in\cP(x)}\P(\overline{A_n})\downarrow 0,
\end{equation}
which is exactly \eqref{capty1}. 

Now, for the case where only one of (i) and (ii) holds, we assume without loss of generality that (i) holds. Let $A^{p,q}_n$ denote the set $A_n$ in the previous case where both (i) and (ii) hold, and $A^p_n$ denote the set $A_n$ in the current case where only (i) holds. Note that $A^{p,q}_n = A^{p,q,1}_n\cup A^{p,q,2}_n$, where
\[
A^{p,q,1}_n := \{\omega\in\Omega^x: B_{\rho^0}=p,\ |\rho^n-\rho^0|\ge \eps\},\quad A^{p,q,2}_n := \{\omega\in\Omega^x: B_{\rho^0}=q,\ |\rho^n-\rho^0|\ge \eps\}.
\]
Observing that $A^p_n = A^{p,q,1}_n$, we obtain in the current case
\[
\sup_{\P\in\cP(x)}\P(|\rho^n-\rho^0|\ge \eps) =  \sup_{\P\in\cP(x)}\P(A^{p,q,1}_n) \le \sup_{\P\in\cP(x)}\P(A^{p,q}_n) \downarrow 0,
\]
where the convergence was established in \eqref{to 0 case 1}. That is, \eqref{capty1} remains valid. 
\end{proof}

\begin{proof}[Proof of \eqref{capty2}]
Fix $0<\eps<q-p$. First, note that if $p_n=p$ and $q_n=q$ for $n$ large enough, then $\rho^n=\rho^0$ on $\Omega^x$ for all $n$ large enough, whence \eqref{capty2} follows trivially. It remains to deal with the case where (i) $p_n$ is strictly increasing, or (ii) $q_n$ is strictly decreasing. 

For any $n\in\N$, define $A_n := \{\omega\in\Omega^x:|B_{\rho^n}-B_{\rho^0}|1_{\{\rho^n<\infty\}}\ge \eps\}$. By \eqref{rho converge}, $(A_n)_{n\in\N}$ is nonincreasing for $n$ large enough and $\bigcap_{n\in\N}A_n=\emptyset$. We first deal with the case where (i) and (ii) both hold. For $n$ large enough such that $\max\{|p_n-p|, |q_n-q|\}<\eps$, observe that $\overline{A_n}\setminus A_n \subseteq F_n:= F^1_n\cup F^2_n$, where
\begin{align*}
F^1_n &:= \{\omega\in\Omega^x:B_{\rho^0}=q,\ B_t\le q_n\ \forall t\in(\rho^0,\rho_{\{p_n\}}),\ \exists s\in(\rho^0,\rho_{\{p_n\}})\ \hbox{s.t.}\ B_s= q_n\},\\
F^2_n &:= \{\omega\in\Omega^x:B_{\rho^0}=p,\ B_t\ge p_n\ \forall t\in(\rho^0,\rho_{\{q_n\}}),\ \exists s\in(\rho^0,\rho_{\{q_n\}})\ \hbox{s.t.}\ B_s= p_n\}. %\quad \forall n\in\N. 
\end{align*}
Note that \eqref{F_n null} holds in the current context, by the same argument below \eqref{F_n null}. Also, by the definitions of $F^1_n$ and $F^2_n$, $F_n\cap F_m=\emptyset$ for all $n<m$, as $p_n$ is strictly increasing and $q_n$ is strictly decreasing. It follows that \eqref{also empty} is still true. Hence, by using Lemma 7 in \cite{DHP11} again, we obtain \eqref{to 0 case 1}, which is exactly \eqref{capty2}.  

For the case where only one of (i) and (ii) holds, we can follow the same argument in the last part of the proof of \eqref{capty1} to conclude that \eqref{capty2} remains valid. 
\end{proof}

{\small
\bibliographystyle{agsm}
\bibliography{refs}
}

\end{document}